%% file: Tomita-Takesaki_modular_theory_new_version.tex
\documentclass[12pt]{article}

\input{preamble.tex}

\begin{document}

\title{\large Tomita-Takesaki Modular Theory vs. Quantum Information Theory}
\author{Lin Zhang\footnote{E-mail: linyz@zju.edu.cn; godyalin@163.com}\\
{\small\it Institute of Mathematics, Hangzhou Dianzi
University, Hangzhou 310018, PR~China}\\[1mm]
Junde Wu\footnote{E-mail: wjd@zju.edu.cn}\\[1mm]
{\small \it Department of Mathematics, Zhejiang University, Hangzhou
310027, PR~China}}
\date{}\maketitle

\mbox{}\hrule\mbox{}
\begin{abstract}

In this paper, we will try to find out the relationship
between separating and cyclic vectors in the theory of von Neumann
algebra and entangled states in the theory of quantum information.
The corresponding physical interpretation is presented as well.

\end{abstract}
\mbox{}\hrule\mbox{}
\tableofcontents

\section{Introduction}

First we recall that some concepts about Tomita-Takesaki modular
theory in the setting of von Neumann algebras. For this parts, the
readers are referred to \cite{Bratteli}.

\subsection{von Neumann algebras}

Let $\cH$ be a Hilbert space and $\lin{\cH}$ the all bounded linear
operator defines on $\cH$. Assume that $\cM$ is a subset of $\lin{\cH}$,
we denote its commutant $\cM'$ by the set of all bounded operators
on $\cH$ commuting with every operator in $\cM$, that is
$$
\cM' \defeq \Set{M'\in\lin{\cH}: [M',M]=M'M-MM'=0\text{\ for all\ }M\in\cM}.
$$
One has
\begin{eqnarray*}
\cM &\subseteq& \cM'' = \cM^{(\mathrm{iv})} = \cM^{(\mathrm{vi})} =
\cdots\\
\cM' &=& \cM''' = \cM^{(\mathrm{v})} = \cM^{(\mathrm{vii})} = \cdots
\end{eqnarray*}

\begin{definition}
A $\ast$-algebra $\cM$ on $\cH$ is said to be a \emph{von Neumann algebra}, if
$$
\cM=\cM''.
$$

The \emph{center} $\cC$ of a von Neumann algebra $\cM$ is
defined by
$$
\cC \defeq \cM\cap\cM'.
$$

A von Neumann algebra is called a \emph{factor} if $\cC=\complex\I$.
\end{definition}

\begin{definition}
If $\cA$ is a subset of $\lin{\cH}$ and $\cX$ is a subset of $\cH$,
let $[\cA\cX]$ denote the closure of the linear span of elements of
the form $Ax$ for all $A\in\cM,x\in\cX$. Let $[\cA\cX]$ also denote
the orthogonal projection onto $[\cA\cX]$.
\end{definition}

\begin{prop}\label{prop:usualtrace}
Let $\trace$ be the usual trace on $\lin{\cH}$, and let $\cC_1$ be
the Banach space of trace-class operators on $\cH$ equipped with the
trace norm $T\mapsto\Tr{\abs{T}}\defeq \norm{T}_1$. Then $\lin{\cH}$ is the dual $\cC^\ast_1$ of $\cC_1$ by the duality:
$$
(A,T)\in\lin{\cH}\times \cC_1\longmapsto \Tr{AT}.
$$
\end{prop}

Let $\{\ket{\xi_n}\}$ and $\{\ket{\eta_n}\}$ be two sequence of vectors in ${\cH}$, such that $\norm{\ket{\xi_n}}^2<+\infty$  and $\sum_n
\norm{\ket{\eta_n}}^2<+\infty$. Then for each $A\in\lin{\cH}$, $$A\rightarrow \sum_n |\Innerm{\xi_n}{A}{\eta_n}|$$ defined a seminorm on $\lin{\cH}$. The locally convex topology on $\lin{\cH}$ is induced by these seminorms is called the $\sigma$-weakly topology.

The $\sigma$-weakly topology of $\lin{\cH}$ is just the $w^{\ast}$ topology induced by $\cC_1$.

\begin{definition}
The space of $\sigma$-weakly continuous linear functionals on
$\lin{\cH}$ is called the \emph{predual} of $\lin{\cH}$ and is
denoted by $\lin{\cH}_\ast$.
\end{definition}

Note that the $\sigma$-weakly topology of $\lin{\cH}$ is just the $w^{\ast}$ topology induced by $\cC_1$, so
$\lin{\cH}_\ast=\cC_1$.

\subsubsection{Normal states and the predual}

\begin{definition}
The \emph{predual} of a von Neumann algebra $\cM$ is the space of
all $\sigma$-weakly continuous linear functionals on $\cM$. It is denoted
by $\cM_\ast$.
\end{definition}

Let $\cM$ be a von Neumann algebra, denote $\cM^\bot=\{T: T\in \cC_1, \Tr{TM}=0, \text{\ for all\ }M\in\cM\}.$ Then we have $\cM_\ast=\lin{\cH}_\ast/\cM^\bot$, and $\cM_\ast^{_\ast}=\cM$.

\begin{definition}
Let $\cM$ be a von Neumann algebra on a Hilbert space $\cH$, and $\omega$ a positive linear functional on $\cM$. We say that $\omega$ is normal if for all increasing nets $\{A_{\alpha}\}$ in $\cM_{+}$ with an upper bound, then $\omega(\sup_{\alpha}A_{\alpha})=\sup_{\alpha}\omega(A_{\alpha})$.
\end{definition}

\begin{remark}
Now we remark here that we can construct an operator (in Dirac
notation)
$$
\sum_n \out{\xi_n}{\eta_n}
$$
when $\sum_n \norm{\ket{\xi_n}}^2<+\infty$  and $\sum_n
\norm{\ket{\eta_n}}^2<+\infty$.

Since
\begin{eqnarray*}
\norm{\sum_n \out{\xi_n}{\eta_n}} &\leqslant& \sum_n\norm{
\out{\xi_n}{\eta_n}} = \sum_n\norm{\ket{\xi_n}}\norm{\ket{\eta_n}}\\
&\leqslant&
\Pa{\sum_n\norm{\ket{\xi_n}}^2}\Pa{\sum_n\norm{\ket{\eta_n}}^2}<+\infty,
\end{eqnarray*}
it follows that $\sum_n \out{\xi_n}{\eta_n}\in\lin{\cH}$. Therefore,
each normal element $\omega\in\cM_\ast$ has a representative in
$\lin{\cH}$:
$$
\omega(M) = \Inner{\sum_n \out{\xi_n}{\eta_n}}{M}_{\bH\bS},
$$
where $\Inner{X}{Y}_{\bH\bS} \defeq \Tr{X^\ast Y}$.
\end{remark}

\begin{prop}
The predual $\cM_\ast$ of a von Neumann algebra $\cM$ is a Banach
space  in the norm of $\cM^\ast$, and $\cM$ is the dual of
$\cM_\ast$ in the duality
$$
(M,\omega)\in \cM\times \cM_\ast \longmapsto \omega(M).
$$
\end{prop}

\begin{remark}
We recall the following identifications:
\begin{eqnarray*}
\ell^\ast_1 = \ell_\infty,\quad L^\ast_1 = L_\infty,\quad \cC^\ast_1
= \lin{\cH}.
\end{eqnarray*}
Thus the predual $\cM_\ast$ of von Neumann algebra $\cM$ can be
viewed as an analog of $\cC_p$-class with $p=1$ in $\lin{\cH}$.
Therefore, we have that if denote $\cM_\infty \equiv \cM$ and $\cM_1
\equiv \cM_\ast$
$$
(\cM_\ast)^\ast = \cM,\quad \text{or}\quad \cM^\ast_1 = \cM_\infty.
$$
In particular, when $\cM=\lin{\cH}$, we have $\lin{\cH}^\ast_1 =
\lin{\cH}_\infty$.
\end{remark}

\begin{prop}
Let $\omega$ be a state on a von Neumann algebra $\cM$ acting on a
Hilbert space $\cH$. Then the following conditions are equivalent:
\begin{enumerate}[(i)]
\item $\omega$ is normal;
\item $\omega$ is $\sigma$-weakly continuous;
\item there exist a density matrix $D_\omega$, that is, a positive trace-class operator
$D_\omega$ on $\cH$ with $\Tr{D_\omega}=1$ such that $\omega(M) =
\Tr{D_\omega M}$ for all $M$ in $\cM$.
\end{enumerate}
\end{prop}

\begin{remark}
We recall that Riesz Representation Theorem describes continuous
functional on a Hilbert space has a vector representative: If $f$ is
a continuous functional on a Hilbert space $\cH$, then there is a
vector $\ket{u_f}\in\cH$ such that
$$
f(\ket{v}) = \iinner{u_f}{v},\quad
\forall \ket{v}\in\cH.
$$
By comparison with Riesz Representation Theorem, we have: For each
normal state $\omega\in\cM_\ast$, it has a representative $D_\omega$
in $\lin{\cH}$ as follows:
$$
\omega(M) = \Inner{D_\omega}{M}_{\bH\bS}.
$$
By the definition of normal element in $\cM_\ast$, there exist a
sequence of vectors $\set{\ket{\psi_n}}$ in $\cH$,
$\sum_n\norm{\ket{\psi_n}}^2<+\infty$, such that
$$
D_\omega = \sum_n \out{\psi_n}{\psi_n}.
$$
Furthermore, setting $\lambda_n\defeq \norm{\ket{\psi_n}}^2>0$ and
$\ket{\psi_n}=\sqrt{\lambda_n}\ket{\phi_n}$ with $\norm{\phi_n}=1$,
we have
$$
D_\omega = \sum_n \lambda_n\out{\phi_n}{\phi_n}.
$$
\end{remark}

\begin{prop}
Let $\cM$ be a von Neumann algebra acting on a Hilbert space $\cH$.
If $\cI$ be a $\sigma$-weakly closed two-sided ideal in $\cM$. Then
there exists a projection $E\in\cM\cap\cM'$ such that $\cI = E\cM
E$.
\end{prop}

\subsubsection{$\sigma$-finite von Neumann algebras}

\begin{definition}
A von Neumann algebra $\cM$, acting on a Hilbert space $\cH$, is
\emph{$\sigma$-finite} if all collections of mutually orthogonal
projections have at most a countable cardinality.
\end{definition}

\begin{definition}
Let $\cM$ be a von Neumann algebra on a Hilbert space $\cH$. A
subset $\cH_0\subseteq\cH$ is \emph{cyclic} for $\cM$ if the set
$\Set{M\ket{u}: M\in\cM, \ket{u}\in\cH_0}$ is dense in $\cH$, i.e.,
$[\cM\cH_0]=\cH$. We say that $\cH_0$ is \emph{separating} for $\cM$
if for any $M\in\cM$, $M\ket{u}=0$ for all $\ket{u}\in\cH_0$ implies $M=0$.
\end{definition}

\begin{prop}
Let $\cM$ be a von Neumann algebra on a Hilbert space $\cH$ and
$\cH_0\subseteq\cH$ a subset. Then $\cH_0$ is cyclic for $\cM$ if
and only if $\cH_0$ is separating for $\cM'$.
\end{prop}

\begin{definition}
Let $\cM$ be a von Neumann algebra on a Hilbert space $\cH$. A
vector $\ket{\Omega}$ is called \emph{cyclic} for $\cM$ if the set
$\Set{M\ket{\Omega}: M\in\cM}$ is dense in $\cH$, i.e., $\Br{\cM\ket{\Omega}} =
\cH$. We say that $\ket{\Omega}\in\cH$ is \emph{separating} for $\cM$ if
for any $M\in\cM$, $M\ket{\Omega}=0$ implies $M=0$.
\end{definition}

\begin{prop}
Let $\cM$ be a von Neumann algebra on a Hilbert space $\cH$ and
$\ket{\Omega}\in\cH$. Then $\ket{\Omega}$ is cyclic for $\cM$ if and only if
$\ket{\Omega}$ is separating for $\cM'$.
\end{prop}

\begin{definition}
A state $\omega$ on a von Neumann algebra $\cM$ is \emph{faithful}
if $\omega(M)>0$ for all nonzero $M\in\cM^+$.
\end{definition}

\begin{prop}
Let $\cM$ be a von Neumann algebra on a Hilbert space $\cH$. The the
following four conditions are equivalent:
\begin{enumerate}[(i)]
\item $\cM$ is $\sigma$-finite;
\item there exists a countable subset of $\cH$ which is separating
for $\cM$;
\item there exists a faithful normal state on $\cM$;
\item $\cM$ is isomorphic with a von Neumann algebra $\pi(\cM)$
which admits a separating and cyclic vector.
\end{enumerate}
\end{prop}

\subsubsection{Tomita-Takesaki modular theory}

\emph{Tomita-Takesaki Modular Theory} has been one of the most
exciting subjects for operator algebras and for its applications to
mathematical physics. We will give here a short introduction to this
theory and state some of its main results.

If von Neumann algebra $\cM$ is a $\sigma$-finite, we may assume
that $\cM$ has a separating and cyclic vector $\ket{\Omega}$. In
Tomita-Takesaki modular theory, one studies systematically the
relation of a von Neumann algebra $\cM$ and its commutant $\cM'$ in
the case where both algebras have a common cyclic vector $\ket{\Omega}$.
The mapping
$$
M\in\cM \longmapsto M\ket{\Omega}\in\cH,
$$
then establishes a one-to-one linear correspondence between $\cM$
and a dense subspace $\cM\ket{\Omega}$ of $\cH$. This correspondence may
be used to transfer algebraic operations on $\cM$ to operations on
$\cM\ket{\Omega}$.

The two anti-linear operators $S_0$ and $F_0$, given by
\begin{eqnarray*}
S_0 M\ket{\Omega} &=& M^\ast\ket{\Omega},\quad \forall M\in\cM,\\
F_0 M'\ket{\Omega} &=& {M'}^\ast\ket{\Omega},\quad \forall M'\in\cM',
\end{eqnarray*}
are both well-defined on the dense domains $D(S_0)=\cM\ket{\Omega}$ and
$D(F_0)=\cM'\ket{\Omega}$.

\begin{prop}
$S_0$ and $F_0$ are closable. And
$$
S_0^\ast = \overline{F}_0,\quad F_0^\ast = \overline{S}_0,
$$
where the bar denotes the closure.
\end{prop}

\begin{definition}
Define $S$ and $F$ as the closures of $S_0$ and $F_0$, respectively,
i.e.,
$$
S = \overline{S}_0,\quad F = \overline{F}_0.
$$
Let $\Delta$ be the unique, positive, self-adjoint operator and $J$
the unique anti-unitary operator occurring in the \emph{polar
decomposition}
$$
S = J\Delta^{\frac12}
$$
of $S$. $\Delta$ is called the \emph{modular operator associated
with the pair} $\Set{\cM,\ket{\Omega}}$ and $J$ is called the
\emph{modular conjugation}.
\end{definition}

\begin{prop}
The following relations are valid:
$$
\left\{\begin{array}{c}
         \Delta = FS \\
         \Delta^{-1}=SF
       \end{array}
\right., \quad \left\{\begin{array}{c}
         S = J\Delta^{\frac12} \\
         F = J\Delta^{-\frac12}
       \end{array}
\right., \quad \left\{\begin{array}{c}
         J^\ast = J \\
         J^2 =\I
       \end{array}
\right., \quad \Delta^{-\frac12} = J\Delta^{\frac12}J.
$$
\end{prop}

\begin{thrm}[Tomita-Takesaki Theorem]
Let $\cM$ be a von Neumann algebra with cyclic and separating vector
$\ket{\Omega}$, and let $\Delta$ be the associated modular operator and
$J$ the associated modular conjugation. It follows that
$$
\left\{\begin{array}{ccl}
         J\cM J &=& \cM', \\
         \Delta^{\mathrm{i}t}\cM\Delta^{-\mathrm{i}t}&=&\cM,\forall
         t\in\real.
       \end{array}
\right.
$$
\end{thrm}

\subsubsection{Self-dual cones and standard forms}

\begin{definition}
The \emph{natural positive cone} $\cP$ associated with the pair
$(\cM,\ket{\Omega})$ is defined as the closure of the set
$$
\Set{Mj(M)\ket{\Omega}: M\in\cM},
$$
where $j: \cM\longmapsto \cM'$ is the anti-linear $\ast$-isomorphism
defined by
$$
j(M) \defeq JMJ, \quad \forall M\in\cM.
$$
\end{definition}

\begin{prop}
The closed subset $\cP\subseteq\cH$ has the following properties:
\begin{enumerate}[(i)]
\item \begin{eqnarray*}
\cP &=& \Br{\Delta^{\frac14}\cM^+\ket{\Omega}} =
\Br{\Delta^{-\frac14}{\cM'}^+\ket{\Omega}}\\
&=&\Br{\Delta^{\frac14}\Br{\cM^+\ket{\Omega}}} =
\Br{\Delta^{-\frac14}\Br{{\cM'}^+\ket{\Omega}}}
\end{eqnarray*}
and hence $\cP$ is a convex cone;
\item $\Delta^{\mathrm{i}t}\cP = \cP$ for all $t\in\real$;
\item if $f$ is a positive-definite function, then
$f(\log\Delta)\cP\subseteq\cP$;
\item if $\ket{\xi}\in\cP$, then $J\ket{\xi} = \ket{\xi}$;
\item if $M\in\cM$, then $Mj(M)\cP\subseteq\cP$.
\end{enumerate}
\end{prop}

\begin{prop}
\begin{enumerate}[(i)]
\item $\cP$ is a self-adjoint cone, i.e., $\cP = \cP^{\vee}$, where
$$
\cP^{\vee} = \Set{\ket{\eta}\in\cH: \iinner{\xi}{\eta}\geqslant0\text{\ for
all\ }\ket{\xi}\in\cP}.
$$
\item $\cP$ is a pointed cone, i.e.,
$$
\cP\cap(-\cP) = \set{0}.
$$
\item If $J\ket{\xi} = \ket{\xi}$, then $\ket{\xi}$ has a unique decomposition $\ket{\xi} = \ket{\xi_1} -
\ket{\xi_2}$, where $\ket{\xi_1},\ket{\xi_2}\in\cP$ and $\ket{\xi_1}\bot \ket{\xi_2}$.
\item $\cH$ is linearly spanned by $\cP$.
\end{enumerate}
\end{prop}

\begin{prop}[Universality of the cone $\cP$]
\begin{enumerate}[(i)]
\item If $\ket{\xi}\in\cP$, then $\ket{\xi}$ is cyclic for $\cM$ if and only if $\ket{\xi}$ is separating for
$\cM$.
\item If $\ket{\xi}\in\cP$ is cyclic and separating, then the modular
conjugation $J_{\ket{\xi}}$ and the natural positive cone $\cP_{\ket{\xi}}$
associated with the pair $(\cM,\ket{\xi})$ satisfy
$$
J_{\ket{\xi}} = J,\quad \cP_{\ket{\xi}} = \cP.
$$
\end{enumerate}
\end{prop}

\begin{thrm}\label{th:positive-cone}
For each $\ket{\xi}\in\cP$, define the normal positive form
$\omega_{\ket{\xi}}\in\cM_{\ast,+}$ by
$$
\omega_\xi(M) = \Innerm{\xi}{M}{\xi},\quad M\in\cM.
$$
It follows that
\begin{enumerate}[(i)]
\item for any $\omega\in\cM_{\ast,+}$, there exists a unique
$\ket{\xi}\in\cP$ such that $\omega=\omega_{\ket{\xi}}$,
\item the mapping $\ket{\xi}\longmapsto \omega_{\ket{\xi}}$ is a homeomorphism
when both $\cP$ and $\cM_{\ast,+}$ are equipped with the norm
topology. Moreover, the following estimates are valid:
$$
\norm{\ket{\xi} - \ket{\eta}}^2 \leqslant \norm{\omega_{\ket{\xi}} - \omega_{\ket{\eta}}}
\leqslant \norm{\ket{\xi} - \ket{\eta}}\norm{\ket{\xi} + \ket{\eta}}.
$$
\end{enumerate}
\end{thrm}

\section{The operator-vector correspondence}

For the operator-vector correspondence \cite{Watrous}, we
distinguish two situations where slight differences occurred in the
corresponding definitions.

\subsection{vec mapping in unipartite operator spaces}
It will be helpful throughout this course to make use of a simple
correspondence between the spaces $\lin{\cX,\cY}$ and $\cY \ot \cX$,
for given complex Euclidean spaces $\cX$ and $\cY$. We define the
mapping
$$
\vec : \lin{\cX,\cY}\longrightarrow \cY \ot \cX
$$
to be the linear mapping that represents a change of bases from the
standard basis of $\lin{\cX,\cY}$ to the standard basis of $\cY \ot
\cX$. Specifically, we define
$$
\vec(E_{\mu,\nu}) = e_\mu \ot e_\nu
$$
for all $\mu\in\Sigma$ and $\nu\in\Gamma$, at which point the
mapping is determined for every $A\in\lin{\cX,\cY}$ by linearity. In
the Dirac notation, this mapping amounts to flipping a bra to a ket:
$$
\col{\out{\mu}{\nu}} = \ket{\mu}\ot\ket{\nu}\equiv
\ket{\mu}\ket{\nu} \equiv \ket{\mu\nu}.
$$
(Note that it is only standard basis elements that are flipped in
this way.)

The $\vec$ mapping is a linear bijection, which implies that every
vector $\ket{u}\in\cY\ot\cX$ uniquely determines an operator
$A\in\lin{\cX,\cY}$ that satisfies $\col{A} = \ket{u}$. It is also an
isometry, in the sense that
$$
\inner{A}{B} = \inner{\col{A}}{\col{B}}
$$
for all $A,B\in\lin{\cX,\cY}$. The following properties of the
$\vec$ mapping are easily verified:
\begin{enumerate}[(i)]
\item For every choice of complex Euclidean spaces $\cX_1, \cX_2, \cY_1$, and $\cY_2$, and every choice of
operators $A\in\lin{\cX_1,\cY_1}, B\in\lin{\cX_2,\cY_2}$, and
$X\in\lin{\cX_2,\cX_1}$, it holds that
\begin{eqnarray}\label{eq:vecidentity1}
(A\ot B)\col{X} = \col{AXB^\t}.
\end{eqnarray}
\item For every choice of complex Euclidean spaces $\cX$ and $\cY$, and every choice of
operators $A,B\in\lin{\cX,\cY}$, the following equations hold:
\begin{eqnarray}\label{eq:vecidentity2}
\Ptr{\cX}{\col{A}\col{B}^\ast} = AB^\ast,
\end{eqnarray}
\begin{eqnarray}\label{eq:vecidentity3}
\Ptr{\cY}{\col{A}\col{B}^\ast} = (B^\ast A)^\t.
\end{eqnarray}
\item For $\ket{u}\in\cX$ and $\ket{v}\in\cY$ we have
\begin{eqnarray}\label{eq:vecidentity4}
\col{\out{u}{v}} = \ket{u}\ot\overline{\ket{v}}.
\end{eqnarray}
This includes the special cases $\col{\ket{u}}= \ket{u}$ and $\col{\bra{v}} =
\overline{\ket{v}}$.
\end{enumerate}

\begin{exam}[The Schmidt decomposition] Suppose $\ket{u}\in\cY\ot\cX$ for given complex Euclidean
spaces $\cX$ and $\cY$. Let $A\in\lin{\cX,\cY}$ be the unique
operator for which $\ket{u} = \col{A}$. There exists a singular value
decomposition
$$
A= \sum^r_{i=1}s_i \out{y_i}{x_i}
$$
of $A$. Consequently
$$
\ket{u} = \col{A} = \col{\sum^r_{i=1}s_i \out{y_i}{x_i}} =
\sum^r_{i=1}s_i\col{\out{y_i}{x_i}} = \sum^r_{i=1}s_i \ket{y_i} \ot
\overline{\ket{x_i}}.
$$
The fact that $\Set{\ket{x_1},\ldots,\ket{x_r}}$ is orthonormal implies that
$\Set{\overline{\ket{x_1}},\ldots,\overline{\ket{x_r}}}$ is orthonormal as well.

We have therefore established the validity of the Schmidt
decomposition, which states that every vector $\ket{u} \in \cY \ot \cX$
can be expressed in the form
$$
\ket{u} = \sum^r_{i=1}s_i \ket{y_i} \ot \ket{z_i}
$$
for positive real numbers $s_1,\ldots,s_r$ and orthonormal sets
$$
\Set{\ket{y_1},\ldots, \ket{y_r}}\subset \cY\quad\text{and}\quad\Set{\ket{z_1},\ldots, \ket{z_r}}\subset
\cX.
$$
\end{exam}

\subsection{vec mapping in multipartite operator spaces}
When the $\vec$ mapping is generalized to multipartite spaces,
caution should be given to the bipartite case (multipartite
situation similarly). Specifically, for given complex Euclidean
spaces $\cX_{A/B}$ and $\cY_{A/B}$,
$$
\vec: \lin{\cX_A \ot \cX_B,\cY_A \ot \cY_B} \longrightarrow \cY_A
\ot \cX_A \ot \cY_B \ot \cX_B
$$
is defined to be the linear mapping that represents a change of
bases from the standard basis of $\lin{\cX_A \ot \cX_B,\cY_A \ot
\cY_B}$ to the standard basis of $\cY_A \ot \cX_A \ot \cY_B \ot
\cX_B$. Concretely,
$$
\col{\out{m}{n} \ot \out{\mu}{\nu}}: = \ket{mn} \ot
\ket{\mu\nu}\equiv \ket{mn\mu\nu},
$$
where $\set{\ket{n}}$ is an orthonormal basis for $\cX_A$ and
$\set{\ket{\nu}}$ is an orthonormal basis for $\cX_B$, while
$\set{\ket{m}}$ is an orthonormal basis for $\cY_A$ and
$\set{\ket{\mu}}$ is an orthonormal basis for $\cY_B$. Analogously,
the mapping is determined for every operator $X \in \lin{\cX_A \ot
\cX_B,\cY_A \ot \cY_B}$ by linearity. Note that if $X = A \ot B$,
where $A \in \lin{\cX_A,\cY_A}$ and $B \in \lin{\cX_B,\cY_B}$, then
$$
\col{A \ot B} = \col{A}\ot \col{B}.
$$

\section{Explicit examples}

\begin{exam}

Let $\cH_d$ be a $d$-dimensional complex Hilbert space. Consider a
von Neumann algebra $\cM \equiv \lin{\cH_d}$. For any
$X\in\lin{\cH_d}$, the following map defined a faithful
representation of von Neumann algebra $\cM$ on a Hilbert space
$\cH\equiv\cH_d\ot\cH_d$:
$$
\pi: X\longmapsto \pi(X) = X\ot\I_d.
$$
Setting $\ket{\Omega} \defeq \col{\I_d}$, we have that $\ket{\Omega}$ is a
separating and cyclic vector in $\cH$ for von Neumann algebra
$\pi(\cM) \equiv \lin{\cH_d}\ot\I_d$. Therefore we can conclude that
von Neumann algebra $\cM$ have a standard representation
$(\pi(\cM),\cH,\ket{\Omega})$.

Consider the Tomita-Takesaki modular theory in
$(\pi(\cM),\cH,\ket{\Omega})$. According to the Tomita-Takesaki modular
theory
$$
S\pi(X)\ket{\Omega} \defeq \pi(X)^\ast\ket{\Omega} = (X^\ast\ot\I_d)\ket{\Omega},\quad
\forall X\in\cM,
$$
which is equivalently described as
$$
S\col{X} = \col{X^\ast},\quad \forall X\in\cM.
$$
If we assume that $K$ is the \emph{complex conjugate operator} and
$P$ is a \emph{swap operator}, then
\begin{eqnarray*}
S\col{X} &=& \col{X^\ast} = \col{(\overline{X})^\t} =
P\col{\overline{X}}\\
 &=& PK\col{X}= KP\col{X},
\end{eqnarray*}
which means that $S = PK = KP$. Similarly, $J= S=F=PK=KP$, therefore
$\Delta=\I$. In quantum physics, $K$ stands for time reversal operation.

\end{exam}

\begin{thrm}

The set of all separating and cyclic vectors in $\cH$ for $\pi(\cM)$
is precisely the set
$$
\Set{\col{A}\in\cH: A\in\cM \text{\ is not singular}}.
$$
\end{thrm}

\begin{proof}
If $A\in\cM$ is not singular, then for any $\pi(X)\in\pi(\cM)$, we
have
$$
\pi(X)\col{A} = 0 \Longleftrightarrow \col{XA}=0 \Longleftrightarrow
XA=0 \Longleftrightarrow X=0.
$$
Thus $\col{A}$ is a separating vector. When $X$ is all over $\cM$,
we have
$$
\pi(\cM)\col{A} = \col{\cM A}= \col{\cM} = \cH,
$$
which implies that $\col{A}$ is a cyclic vector.

Now suppose that $\ket{\psi}\in\cH$ is a separating and cyclic
vector for $\pi(\cM)$. Then there exists an operator $B_\psi\in\cM$
such that $\ket{\psi}=\col{B_\psi}$. If $B_\psi$ is singular, then
$\cM B_\psi$ is a proper left ideal of $\cM$. Thus
$\pi(\cM)\ket{\psi} \neq\cH$ and there exists $X_1\neq X_2$ such
that $X_1B_\psi = X_2B_\psi$. That is, $\pi(X_1) \ket{\psi}=\pi(X_2)
\ket{\psi}$. Therefore $\ket{\psi}=\col{B_\psi}$ is not a separating
and cyclic vector for singular operator $B_\psi$.
\end{proof}

\begin{remark}
We recall that the \emph{Schmidt rank} of pure state
$\ket{\psi}\in\cH$ is defined by
$$
\rS\rR(\ket{\psi}) \defeq \rank(B_\psi), \quad \ket{\psi} =
\col{B_\psi}.
$$
Hence the above result can be described equivalently as:

\textbf{Claim:} $\ket{\psi}\in\cH$ is a separating and cyclic vector
for $\pi(\cM)$ if and only if $\rS\rR(\ket{\psi}) = d$.

In some sense,  separating and cyclic vectors stands for quantum
states of most entanglement of measure.

If $\omega$ is a state on $\cM$, then there exist density matrix
$D_\omega\in\lin{\cH_d}$ such that
$$
\omega(M) = \Tr{D_\omega M} = \Inner{D_\omega}{M}_{\bH\bS},\quad
M\in\lin{\cH_d}.
$$
It is known that $\omega$ is faithful if and only if $D_\omega$ is
not singular. Since $\dim(\cH_d) = d <+\infty$, it follows that all
states on $\cM$ are normal.

Consider the normalized vector $\ket{\Omega}\defeq
\col{\sqrt{D_\omega}}\in\cH$ for faithful normal state $\omega$. It
is easily seen that
$$
\omega(M) = \Innerm{\Omega}{\pi(M)}{\Omega}.
$$
$\ket{\Omega}$ is a separating and cyclic vector $\pi(\cM)$. In terms of
the language of quantum information theory, $\ket{\Omega}$ is a
purification of density matrix $D_\omega$ in $\cH$. Thus there is a
connection between the standard representation of von Neumann
algebra with a faithful normal state and \emph{purification} of density
matrix:

Given a faithful normal state $\omega$ on von Neumann algebra $\cM$.
Then the standard representation of $\cM$ is
$(\pi(\cM),\cH,\ket{\Omega})$, where $\ket{\Omega}=\col{\sqrt{D_\omega}}$ is a
purification of density matrix $D_\omega$ which is not singular.
\end{remark}

\begin{exam}[Unification of finite or countable infinite situation, \cite{Ali,Bagarello}]
A simple example of the Tomita-Takesaki theory and its related KMS
states can be built on the space of \emph{Hilbert-Schmidt operators}
on a Hilbert space. The set of Hilbert-Schmidt operators is itself a
Hilbert space, and there are two preferred algebras of operators on
it, which carry the modular structure.

Let $\sH$ be a (complex, separable) Hilbert space of dimension $N$
(finite or infinite) and $\set{\ket{\psi_i}}^N_{i=1}$ an orthonormal
basis of it. We denote by $\cC_2$ the space of all Hilbert-Schmidt
operators on $\sH$ ($\cC_2\subset\lin{\sH}$). This is a Hilbert
space with scalar product:
$\Pa{\cC_2,\Inner{\cdot}{\cdot}_{\bH\bS}}$
$$
\Inner{X}{Y}_{\bH\bS} = \Tr{X^\ast Y}.
$$
The vectors (an element of $\cC_2$ is called vector although it is
operator on $\sH$),
$$
\Set{E_{ij}=\out{\psi_i}{\psi_j}: i,j=1,2,\ldots,N}
$$
form an orthonormal basis of $\cC_2$,
$$
\Inner{E_{ij}}{E_{kl}} = \delta_{ik}\delta_{jl}.
$$
In particular, the vectors,
$$
E_{ii} = \out{\psi_i}{\psi_i},
$$
are one dimensional projection operators on $\sH$. In what follows
$\I$ will denote the identity operator on $\sH$ and $\I_2$ that on
$\cC_2$ (in later notation: $\I_2 = \I\boxtimes\I$).

All bounded linear operator acting on $\cC_2$ (i.e., linear
super-operators in $\trans{\sH}$) are denoted by $\lin{\cC_2}$. We
identify a special class of linear operators on $\cC_2$, denoted by
$A\boxtimes B\in\lin{\cC_2}, A,B\in\lin{\sH}$, which act on a vector
$X\in\cC_2$ in the manner:
$$
A\boxtimes B(X)
\defeq AXB^\ast.
$$
Using the scalar product in $\cC_2$, we see that
\begin{enumerate}[(i)]
\item $(A\boxtimes B)^\ast = A^\ast\boxtimes B^\ast$,
\item $(A_1\boxtimes
B_1)(A_2\boxtimes B_2) = A_1A_2\boxtimes B_1B_2$.
\end{enumerate}
Indeed,
\begin{eqnarray*}
\Inner{(A\boxtimes B)^\ast (Y)}{X}_{\bH\bS} &=&
\Inner{Y}{(A\boxtimes B)(X)}_{\bH\bS} = \Tr{Y^\ast AXB^\ast}\\
&=& \Tr{B^\ast Y^\ast AX} = \Inner{\Pa{B^\ast Y^\ast
A}^\ast}{X}_{\bH\bS} \\
&=& \Inner{A^\ast Y B}{X}_{\bH\bS} =\Inner{A^\ast\boxtimes
B^\ast(Y)}{X}_{\bH\bS},
\end{eqnarray*}
which implies that $(A\boxtimes B)^\ast = A^\ast\boxtimes B^\ast$.
Similar reasoning goes for $(A_1\boxtimes B_1)(A_2\boxtimes B_2) =
A_1A_2\boxtimes B_1B_2$.

There are two special von Neumann algebras which can be built out of
these operators. These are,
$$
\cA_l \defeq \Set{A_l = A\boxtimes\I: A\in\lin{\sH}},\quad \cA_r
\defeq \Set{A_r = \I\boxtimes A: A\in\lin{\sH}}.
$$
As a matter of fact, $A_l$ is a left regular representation of $A$
or a left multiplication by $A$; $A_r$ is a right regular
representation of $A^\ast$ or a right multiplication by $A^\ast$.
For any $A,B\in\lin{\sH}$, we have
$$
A_lB_r = B_r A_l,\quad [A_l,B_r]=0.
$$
In fact, for any $X\in\cC_2$,
\begin{eqnarray*}
A_lB_r(X) &=& (A\boxtimes\I)(\I\boxtimes B)(X) =
(A\boxtimes\I)(XB^\ast)\\
&=& AXB^\ast = (\I\boxtimes B)(AX) = (\I\boxtimes
B)(A\boxtimes\I)(X)\\
&=& B_rA_l(X).
\end{eqnarray*}
They are mutual commutants and both are factors:
$$
\Pa{\cA_l}' = \cA_r,\quad \Pa{\cA_r}' = \cA_l,\quad \cA_l\cap\cA_r =
\complex\I_2.
$$
Consider now the operator $J:\cC_2\longrightarrow\cC_2$, whose
action on the vectors $E_{ij}$ is given by
$$
JE_{ij} \defeq E_{ji}\Longrightarrow J^2 = \I_2,\quad
J(\out{\phi}{\psi}) = \out{\psi}{\phi},\quad\forall
\ket{\phi},\ket{\psi}\in\sH.
$$
This operator is anti-unitary, and since
\begin{eqnarray*}
[J(A\boxtimes \I)J]E_{ij} &=&  J(A\boxtimes \I)E_{ji} = J(AE_{ji})\\
&= & J(A\out{\psi_j}{\psi_i}) = \out{\psi_i}{\psi_j}A^\ast =
(\I\boxtimes A)E_{ij},
\end{eqnarray*}
we immediately get
$$
J\cA_l J = \cA_r.
$$

$\bullet$ \textbf{A KMS state.}

Let $\Set{\lambda_i}_{i=1}^N(N\leqslant+\infty)$ be a sequence of
non-zero, positive numbers, satisfying, $\sum^N_{i=1}\lambda_i=1$.
Let
\begin{eqnarray}
\mathbf{\Omega} \defeq \sum^N_{i=1}\sqrt{\lambda_i}E_{ii}\in\cC_2.
\end{eqnarray}
We note the following properties of $\mathbf{\Omega}$.
\begin{enumerate}[(i)]
\item $\mathbf{\Omega}$ defines a vector state $\omega$ on the von Neumann
algebra $\cA_l$. This follows from the fact that for any
$A\boxtimes\I\in\cA_l$, we may define the state $\omega$ on $\cA_l$
by
\begin{eqnarray}\label{eq:densitymatrix}
\omega(A\boxtimes\I) \defeq
\Inner{\mathbf{\Omega}}{(A\boxtimes\I)\mathbf{\Omega}}_{\bH\bS} =
\Tr{\mathbf{\Omega}^\ast A\mathbf{\Omega}} = \Tr{D_\omega A},\quad
D_\omega = \sum^N_{i=1}\lambda_i E_{ii},
\end{eqnarray}
thus $\mathbf{\Omega}= D^{\frac12}_\omega$.
\item The state $\omega$ is faithful and normal. Normality follows from the last equality
in Eq.~(\ref{eq:densitymatrix}) and the fact that $D_\omega$ is a
density matrix. To check for faithfulness, note that for any
$A\boxtimes\I\in\cA_l$,
$$
\omega((A\boxtimes\I)^\ast(A\boxtimes\I)) = \omega(A^\ast
A\boxtimes\I)) = \Tr{D_\omega A^\ast A} = \sum^N_{i=1}\lambda_i
\norm{A\ket{\psi_i}}^2,
$$
from which it follows that
$\omega((A\boxtimes\I)^\ast(A\boxtimes\I)) = 0$ if and only if $A=0$
(since the $\ket{\psi_i}$ are an orthonormal basis set and the
$\lambda_i>0$), hence if and only if $A\boxtimes\I = 0$.
\item The vector $\mathbf{\Omega}$ is cyclic and separating for
$\cA_l$: $\Br{\cA_l\mathbf{\Omega}}=\cC_2$. Indeed, cyclicity
follows from the fact that if $X\in\cC_2$ is orthogonal to all
$(A\boxtimes\I)\mathbf{\Omega},A \in\lin{\sH}$, then
$$
\Inner{X}{(A\boxtimes\I)\mathbf{\Omega}}_{\bH\bS} = \Tr{X^\ast
A\mathbf{\Omega}} =
\sum^N_{i=1}\sqrt{\lambda_i}\Innerm{\psi_i}{X^\ast A}{\psi_i}=
0,\quad\forall A\in\lin{\cH}.
$$
Taking $A=E_{kl}$, we easily get from the above equality,
$\Innerm{\psi_l}{X^\ast}{\psi_k}=0$ and since this holds for all $k,
l$, we get $X = 0$. In the same way, $\mathbf{\Omega}$ is also
cyclic for $\cA_r$, hence separating for $\cA_l$, i.e.,
$(A\boxtimes\I)\mathbf{\Omega}=(B\boxtimes\I)\mathbf{\Omega}\Longleftrightarrow
A\boxtimes\I=B\boxtimes\I$.
\end{enumerate}
We shall show in the sequel that the state $\omega$ constructed
above is indeed a KMS state for a particular choice of $\lambda_i$.

$\bullet$ \textbf{Time evolution and modular automorphism.}

We now construct a time evolution $\sigma^\omega_t(t\in\real)$, on
the algebra $\cA_l$, using the state $\omega$, with respect to which
it has the KMS property, for fixed $\beta>0$,
$$
\omega(A_l\sigma^\omega_{t+\mathrm{i}\beta}(B_l)) =
\omega(\sigma^\omega_t(B_l)A_l),\quad \forall A_l,B_l\in\cA_l,
$$
and moreover the function,
$$
F_{A_l,B_l}(z) \defeq \omega(A_l\sigma^\omega_z(B_l)),
$$
is analytic in the strip $\Set{z\in\complex:
0<\mathrm{Im}(z)<\beta}$ and continuous on its boundaries. We start
by defining the operators,
$$
\bP_{ij} \defeq E_{ii}\boxtimes E_{jj}.
$$
Clearly $\bP_{ij}$ are projection operators on the Hilbert space
$\cC_2$:
$$
\left\{\begin{array}{c}
         \bP_{ij}^\ast = \bP_{ij}, \\
         \bP_{ij}^2 = \bP_{ij}.
       \end{array}
\right.
$$
Indeed,
\begin{eqnarray*}
\bP_{ij}^\ast &=& (E_{ii}\boxtimes E_{jj})^\ast =
E_{ii}^\ast\boxtimes E_{jj}^\ast = E_{ii}\boxtimes E_{jj} = \bP_{ij},\\
\bP_{ij}^2 &=& (E_{ii}\boxtimes E_{jj})^2 = E^2_{ii}\boxtimes
E^2_{jj} = E_{ii}\boxtimes E_{jj} = \bP_{ij}.
\end{eqnarray*}
Using $D_\omega$ and for a fixed $\beta>0$, define the operator
$H_\omega$ as:
$$
D_\omega \defeq e^{-\beta H_\omega} \Longrightarrow H_\omega =
-\frac1\beta\ln D_\omega = -\frac1\beta
\sum^N_{i=1}(\ln\lambda_i)E_{ii}.
$$
Clearly $\Br{D_\omega,H_\omega}=0$. Next we define the operators:
$$
H^l_\omega \defeq H_\omega\boxtimes\I,\quad H^r_\omega \defeq
\I\boxtimes H_\omega,\quad \bH_\omega \defeq H^l_\omega - H^r_\omega
$$
Since $\sum^N_{i=1}E_{ii}=\I$, we may also write
$$
H^l_\omega = -\frac1\beta \sum^N_{i,j=1}(\ln\lambda_i)\bP_{ij},\quad
H^r_\omega = -\frac1\beta \sum^N_{i,j=1}(\ln\lambda_j)\bP_{ij}.
$$
Thus
$$
\bH_\omega = -\frac1\beta
\sum^N_{i,j=1}\Pa{\ln\frac{\lambda_i}{\lambda_j}}\bP_{ij}.
$$
Using the operator:
$$
\Delta_\omega \defeq \sum^N_{i,j=1}\Pa{\frac{\lambda_i}{\lambda_j}}
\bP_{ij} = e^{-\beta\bH_\omega},
$$
we define a time evolution operator on $\cC_2$:
$$
e^{\mathrm{i}\bH_\omega t} =
\Delta_\omega^{-\frac{\mathrm{i}t}{\beta}}\quad (t\in\real),
$$
and we note that, for any $X\in\cC_2$,
\begin{eqnarray*}
e^{\mathrm{i}\bH_\omega t} (X) &=&
\sum^N_{i,j=1}\Pa{\frac{\lambda_i}{\lambda_j}}^{-\frac{\mathrm{i}t}{\beta}}
\bP_{ij}(X)\\
&=&
\Br{\sum^N_{i=1}\lambda_i^{-\frac{\mathrm{i}t}{\beta}}E_{ii}}\boxtimes
\Br{\sum^N_{j=1}\lambda_j^{-\frac{\mathrm{i}t}{\beta}}E_{jj}}(X)\\
&=& e^{\mathrm{i} H_\omega t}Xe^{-\mathrm{i} H_\omega t},
\end{eqnarray*}
so that
$$
e^{\mathrm{i}\bH_\omega t} = e^{\mathrm{i}H_\omega t}\boxtimes
e^{\mathrm{i}H_\omega t}.
$$
It is clearly that $\mathbf{\Omega}$ commutes with $H_\omega$ and
hence that it is invariant under this time evolution:
$$
e^{\mathrm{i}\bH_\omega t}(\mathbf{\Omega}) = e^{\mathrm{i}H_\omega
t}\mathbf{\Omega} e^{-\mathrm{i}H_\omega t} = \mathbf{\Omega}.
$$
Finally, using $e^{\mathrm{i}\bH_\omega t}(\mathbf{\Omega})$ we
define the time evolution $\sigma^\omega$ on the algebra $\cA_l$, in
the manner:
$$
\sigma^\omega_t(A_l) = e^{\mathrm{i}\bH_\omega t}A_l
e^{-\mathrm{i}\bH_\omega t},\quad\forall A_l\in\cA_l.
$$
Writing $A_l = A\boxtimes\I,A\in\lin{\sH}$, and using the
composition law, we see that
$$
e^{\mathrm{i}\bH_\omega t}A_l e^{-\mathrm{i}\bH_\omega t} =
\Br{e^{\mathrm{i}H_\omega t}A e^{-\mathrm{i}H_\omega t}}\boxtimes\I,
$$
so that
\begin{eqnarray*}
\omega(\sigma^\omega_t(A_l)) &=& \Tr{D_\omega e^{\mathrm{i}H_\omega
t}A e^{-\mathrm{i}H_\omega t}} = \Tr{e^{-\mathrm{i}H_\omega
t}D_\omega e^{\mathrm{i}H_\omega t}A } \\
&=& \Tr{D_\omega A} = \omega(A_l),
\end{eqnarray*}
since
$D_\omega$ and $H_\omega$ commute. Thus, the state $\omega$ is
invariant under the time evolution $\sigma^\omega$.

To obtain the KMS condition, we first note that, with $A_l =
A\boxtimes\I$ and $B_l = B\boxtimes\I$,
$$
A_l \sigma^\omega_t(B_l) = \Br{Ae^{\mathrm{i}H_\omega t}B
e^{-\mathrm{i}H_\omega t}}\boxtimes\I.
$$
Hence,
\begin{eqnarray*}
F_{A_l,B_l}(t) &=& \omega(A_l\sigma^\omega_t(B_l)) = \Tr{D_\omega
Ae^{\mathrm{i}H_\omega t}B e^{-\mathrm{i}H_\omega t}} \\
&=& \Tr{e^{-\mathrm{i}H_\omega t}D_\omega Ae^{\mathrm{i}H_\omega t}B
}=\Tr{D_\omega e^{-\mathrm{i}H_\omega t}A e^{\mathrm{i}H_\omega
t}B},
\end{eqnarray*}
the last equality following from the commutativity of $D_\omega$ and
$H_\omega$. Thus, since $D_\omega=e^{-\beta H_\omega}$, that is,
$D_\omega e^{\beta H_\omega}=\I$. Thus
\begin{eqnarray*}
F_{A_l,B_l}(t+\mathrm{i}\beta) &=& \Tr{D_\omega
e^{-\mathrm{i}H_\omega t}e^{\beta H_\omega}Ae^{\mathrm{i}H_\omega
t}e^{-\beta H_\omega}B} \\
&=& \Tr{D_\omega e^{\beta H_\omega}e^{-\mathrm{i}H_\omega
t}Ae^{\mathrm{i}H_\omega t}e^{-\beta H_\omega}B}
\\
&=&\Tr{e^{-\mathrm{i}H_\omega t}Ae^{\mathrm{i}H_\omega t} D_\omega
B}= \Tr{e^{\mathrm{i}H_\omega t} D_\omega Be^{-\mathrm{i}H_\omega
t}A}\\
&=& \Tr{D_\omega e^{\mathrm{i}H_\omega t}Be^{-\mathrm{i}H_\omega
t}A},
\end{eqnarray*}
so that
$$
\omega(A_l\sigma^\omega_{t+\mathrm{i}\beta}(B_l)) = \Tr{D_\omega
e^{\mathrm{i}H_\omega t}B e^{-\mathrm{i}H_\omega t}A}
=\omega(\sigma^\omega_t(B_l)A_l),
$$
which is the KMS condition.

$\bullet$ \textbf{The anti-linear operator $S_\omega$.}

We now analyze the anti-linear operator $S_\omega:
\cC_2\longrightarrow \cC_2$, which acts as
$$
S_\omega (A_l\mathbf{\Omega}) = A^\ast_l\mathbf{\Omega},\quad
\forall A_l\in\cA_l.
$$
Taking $A_l = A\boxtimes\I$,
$$
S_\omega (A_l\mathbf{\Omega}) = A^\ast_l\mathbf{\Omega},\quad
\forall A_l\in\cA_l\Longleftrightarrow S_\omega (A\mathbf{\Omega}) =
A^\ast\mathbf{\Omega},\quad \forall A\in\lin{\sH}.
$$
Moreover, we may write,
$$
S_\omega (A\mathbf{\Omega}) = A^\ast\mathbf{\Omega} \Longrightarrow
\sum^N_{i=1} \sqrt{\lambda_i}S_\omega (AE_{ii}) = \sum^N_{i=1}
\sqrt{\lambda_i}A^\ast E_{ii}.
$$
Taking $A=E_{kl}$ and using $E_{kl} E_{ii} = \delta_{li}E_{ki}$, we
then get
$$
\sqrt{\lambda_l}S_\omega (E_{kl})  =
\sqrt{\lambda_k}E_{lk}\Longrightarrow S_\omega (E_{kl})  =
\sqrt{\frac{\lambda_k}{\lambda_l}}E_{lk}.
$$
Since any $A\in\lin{\sH}$ can be written as $A=\sum^N_{i,j=1}
a_{ij}E_{ij}$, where $a_{ij} = \Innerm{\psi_i}{A}{\psi_j}$, and
furthermore, since $\bP_{ij}(E_{kl}) =
\delta_{ik}\delta_{jl}E_{ij}$, we obtain
$$
S_\omega = J\Delta^{\frac12}_\omega,
$$
which in fact, also gives the polar decomposition of $S_\omega$.

Thus, we could have obtained the time evolution automorphisms
$\sigma^\omega_t(t\in\real)$, by analyzing the anti-linear operator
$S_\omega$, (since $S^\ast_\omega S_\omega = \Delta_\omega$)
directly. Also, we see that the modular operator simply defines the
\emph{Gibbs state} corresponding to the Hamiltonian $\bH_\omega$.

$\bullet$ \textbf{The centralizer.}

The centralizer of $\cA_l$, with respect to the state $\omega$, is
the von Neumann algebra,
$$
\cM_\omega = \Set{B_l\in\cA_l: \omega(\Br{B_l,A_l})=0,\forall
A_l\in\cA_l}.
$$
Let us determine this von Neumann algebra. Writing
$A_l=A\boxtimes\I,B_l=B\boxtimes\I$, the commutator, $\Br{B_l,A_l}=
(AB-BA)\boxtimes\I$. Hence
$$
\omega(\Br{B_l,A_l}) = \Tr{D_\omega(AB-BA)}.
$$
Thus, in order for the above expression to vanish, we must have,
$$
\sum^N_{i=1}\lambda_i \Innerm{\psi_i}{AB}{\psi_i} =
\sum^N_{i=1}\lambda_i \Innerm{\psi_i}{BA}{\psi_i},\quad \forall
A\in\lin{\cH}.
$$
Taking $A=\out{\psi_k}{\psi_l}$, this gives,
$$
\lambda_k\Innerm{\psi_l}{B}{\psi_k} =
\lambda_l\Innerm{\psi_l}{B}{\psi_k},\quad \forall k,l=1,\ldots,N,
$$
and since in general, $\lambda_k\neq\lambda_l$, this implies that
$\Innerm{\psi_l}{B}{\psi_k}=0$ whenever $k\neq l$. Thus, $B$ is of
the general form $B = \sum^N_{i=1}b_iE_{ii},b_i\in\complex$. In
other words, the centralizer $\cM_\omega$ is generated by the
projectors $E^l_{ii} = E_{ii}\boxtimes\I,i=1,\ldots,N$, which are
\emph{minimal} (i.e., they do not contain projectors onto smaller
subspaces) in $\cA_l$. Alternatively, we may write, $\cM_\omega =
\Set{H^l_\omega}''$, where $H^l_\omega$ is the Hamiltonian defined
above, so that it is an \emph{atomic, commutative} von Neumann
algebra.
\end{exam}

\section{Araki relative modular theory}

Consider a von Neumann algebra $\cM$ in its standard form. If $\cM$
has the standard form $(\cM,\cH,J,\cP)$, then $\cM$ acts on the
Hilbert space $\cH$, $J$ is the modular conjugation, and $\cP$ is a
natural positive cone in $\cH$ such that every faithful normal state
$\omega$ has a unique vector representative $\ket{\Omega}$ in $\cP$ which
is cyclic and separating for $\cM$. Given another normal state
$\phi$, the \emph{densely defined quadratic form}
\begin{eqnarray}
A\ket{\Omega}\mapsto \phi(AA^\ast),\quad \forall A\in\cM
\end{eqnarray}
is \emph{closable} and there exists an associated positive
self-adjoint operator $\Delta$. It is characterized by the following
properties. $\cM\ket{\Omega}$ is a \emph{core} for $\Delta^{\frac12}$ and
$$
\norm{\Delta^{\frac12}A\ket{\Omega}}^2 = \phi(AA^\ast).
$$
The $\Delta$ was called by Araki the \emph{relative modular
operator} \cite{Araki} of $\phi$ and $\omega$ and it is usually
denoted by $\Delta(\phi/\omega)$ or $\Delta_{\phi,\omega}$.
Equivalently, $\Delta_{\phi,\omega}$ is obtained from the polar
decomposition of the \emph{closure} $S_{\phi,\omega}$ of the
conjugate linear operator
$$
A\ket{\Omega}\mapsto A^\ast\ket{\Phi},
$$
where $\ket{\Phi}$ is the vector representative of $\phi$ from $\cP$.
Namely,
$$
S_{\phi,\omega} = J\Delta^{\frac12}_{\phi,\omega}.
$$
The operators $J,\Delta_{\omega,\omega}$ and $\sigma^\omega_t$ are
the standard ingredients of the Tomita-Takesaki modular theory with
respect to $\omega$ or $\ket{\Omega}$. The modular group of $\omega$ is a
one-parameter group of automorphisms of $\cM$ and it looks like
\begin{eqnarray}
\sigma^\omega_t(A) =
\Delta^{\mathrm{i}t}_{\omega,\omega}A\Delta^{-\mathrm{i}t}_{\omega,\omega}.
\end{eqnarray}
Another \emph{Radon-Nikodym derivative-like} object for comparison
of two states is the \emph{Radon-Nikodym cocyle} discovered by
Connes \cite{Connes}. If $\phi$ is a faithful normal state, then
\begin{eqnarray}
[D\phi,D\omega]_t \defeq
\Delta^{\mathrm{i}t}_{\phi,\omega}\Delta^{-\mathrm{i}t}_{\omega,\omega}\equiv
U_t
\end{eqnarray}
is a $\sigma^\omega_t$-cocycle and
\begin{eqnarray}
\sigma^\phi_t = U_t \sigma^\omega_t U^\ast_t.
\end{eqnarray}

\subsection{Functional calculus for a class of super-operators}

We introduce two linear super-operators \cite{Jencova} on the space
$M_d(\complex)$ of $d\times d$ matrices. Left multiplication by $A$
is denoted by $\mathbb{L}_A$ and defined as
$$
\mathbb{L}_A(X) \defeq AX;
$$
right multiplication by $B$ is denoted $\mathbb{R}_B$ and defined as
$$
\mathbb{R}_B(X) \defeq XB.
$$
These super-operators are associated with the relative modular
operator
$$
\Delta_{A,B} = \mathbb{L}_A\mathbb{R}^{-1}_B
$$
introduced by Araki in a far more general context. They have the
following properties:
\begin{enumerate}[(i)]
\item The super-operators $\mathbb{L}_A,\mathbb{R}_B$ commute, i.e. $[\mathbb{L}_A,\mathbb{R}_B]=0$ since
$$
\mathbb{L}_A\mathbb{R}_B(X) = AXB = \mathbb{R}_B\mathbb{L}_A(X)
$$
even when $A$ and $B$ do not commute, i.e. $[A,B]\neq0$.
\item $\mathbb{L}_A$ and $\mathbb{R}_A$ are invertible if and only
if $A$ is non-singular, in which case
$$
\mathbb{L}^{-1}_A = \mathbb{L}_{A^{-1}}\quad\text{and}\quad
\mathbb{R}^{-1}_A = \mathbb{R}_{A^{-1}}.
$$
\item When $A$ is self-adjoint, $\mathbb{L}_A$ and $\mathbb{R}_A$ are both self-adjoint with respect to the
Hilbert-Schmidt inner product $\inner{A}{B}_{\bH\bS} \defeq
\Tr{A^\ast B}$.
\item When $A\geqslant 0$, the super-operators $\mathbb{L}_A$ and
$\mathbb{R}_A$ are positive semi-definite, i.e.
$$
\inner{X}{\mathbb{L}_A(X)}_{\bH\bS} = \Tr{X^\ast AX}\geqslant0,\quad
\inner{X}{\mathbb{R}_A(X)}_{\bH\bS} = \Tr{X^\ast XA} =
\Tr{XAX^\ast}\geqslant0.
$$
\item When $A\geqslant0$, then
$$
(\mathbb{L}_A)^\alpha = \mathbb{L}_{A^\alpha},\quad
(\mathbb{R}_A)^\alpha = \mathbb{R}_{A^\alpha}
$$
for all $\alpha\geqslant0$. If $A>0$, this extends to all real
$\alpha$. More generally,
$$
f(\mathbb{L}_A) = \mathbb{L}_{f(A)}
$$
for all $f:(0,+\infty)\to(-\infty,+\infty)$.
\end{enumerate}

\subsection{Version of super-operator representation}

Suppose that $\mathbf{\Omega}$ and $\mathbf{\Phi}$ are
\emph{separating} and \emph{cyclic} vectors, induced by faithful
normal states $\omega$ and $\phi$, respectively, in $\cC_2$ for
$\cA_l$. Then there exist two \emph{non-singular} density operators
$D_\omega,D_\phi\in\cC_2$ such that
$$
\mathbf{\Omega} = D_\omega^{\frac12},\quad \mathbf{\Phi} =
D_\phi^{\frac12},
$$
According the Araki relative modular theory, we have that for any
$X_l\in\cA_l$ and $Y_r\in\cA_r$,
\begin{eqnarray}
\left\{\begin{array}{ccc}
         S_{\phi,\omega}(X_l\mathbf{\Omega}) &=& X^\ast_l\mathbf{\Phi}, \\
         F_{\phi,\omega}(Y_r\mathbf{\Omega}) &=&
         Y^\ast_r\mathbf{\Phi}.
       \end{array}
\right.
\end{eqnarray}
Both expressions are equivalent to
\begin{eqnarray}
\left\{\begin{array}{ccc}
S_{\phi,\omega}\Pa{XD_\omega^{\frac12}} &=& X^\ast D_\phi^{\frac12}, \\
F_{\phi,\omega}\Pa{D_\omega^{\frac12}Y} &=& D_\phi^{\frac12}Y^\ast,
\end{array}
\right.
\end{eqnarray}
for any $X,Y\in\lin{\sH}$. Thus if the dimension of the underlying
Hilbert space $\sH$ satisfies that $\dim(\sH)<+\infty$, then
\begin{eqnarray}
\left\{\begin{array}{ccc}
S_{\phi,\omega}(A) &=& D_\omega^{-\frac12}A^\ast D_\phi^{\frac12}, \\
F_{\phi,\omega}(B) &=& D_\phi^{\frac12}B^\ast D_\omega^{-\frac12},
\end{array}
\right.
\end{eqnarray}
for any $A,B\in\cC_2$.
$$
\Delta_{\phi,\omega} = F_{\phi,\omega}S_{\phi,\omega} = D_\phi
\boxtimes D^{-1}_\omega,
$$
which implies that
\begin{eqnarray}
\left\{\begin{array}{rcl}
J_{\phi,\omega}X &=& X^\ast,\\
\Delta^{\frac12}_{\phi,\omega}\mathbf{\Omega} &=&
\mathbf{\Phi},\\\Delta^{\mathrm{i}t}_{\phi,\omega} &=&
D^{\mathrm{i}t}_\phi \boxtimes D^{-{\mathrm{i}t}}_\omega.
\end{array}
\right.
\end{eqnarray}

\subsection{Version of vector representation}

Suppose that $\ket{\Omega}$ and $\ket{\Phi}$ are separating and cyclic vectors,
induced by faithful normal states $\omega$ and $\phi$, respectively,
in $\cH\equiv\cH_d\ot\cH_d$ for $\pi(\cM)\equiv \cM\ot\I_d$ with
$\cM = \lin{\cH_d}$. Then there exist two non-singular density
operators $D_\omega,D_\phi\in\lin{\cH_d}$ such that their
purifications are $\ket{\Omega} = \col{D_\omega^{\frac12}}$ and $\ket{\Phi} =
\col{D_\phi^{\frac12}}$. According the Araki relative modular
theory, we have that for any $X,Y\in\lin{\cH_d}$,
\begin{eqnarray}
\left\{\begin{array}{ccc}
S_{\phi,\omega}(X\ot\I_d)\ket{\Omega} &=& (X^\ast\ot\I_d)\ket{\Phi}, \\
F_{\phi,\omega}(\I_d\ot Y)\ket{\Omega} &=& (\I_d\ot Y^\ast)\ket{\Phi}.
\end{array}
\right.
\end{eqnarray}
Both expressions are equivalent to
\begin{eqnarray}
\left\{\begin{array}{ccc}
S_{\phi,\omega}\col{XD_\omega^{\frac12}} &=& \col{X^\ast D_\phi^{\frac12}}, \\
F_{\phi,\omega}\col{D_\omega^{\frac12}Y} &=&
\col{D_\phi^{\frac12}Y^\ast},
\end{array}
\right.
\end{eqnarray}
for any $X,Y\in\lin{\cH_d}$. Thus
\begin{eqnarray}
\left\{\begin{array}{ccc}
S_{\phi,\omega}\col{X} &=& \col{D_\omega^{-\frac12}X^\ast D_\phi^{\frac12}}, \\
F_{\phi,\omega}\col{Y} &=& \col{D_\phi^{\frac12}Y^\ast
D_\omega^{-\frac12}},
\end{array}
\right.
\end{eqnarray}
for any $X,Y\in\lin{\cH_d}$.
$$
\Delta_{\phi,\omega} = FS = D_\phi \ot \Pa{D^{-1}_\omega}^\t,
$$
which implies that
\begin{eqnarray}
\left\{\begin{array}{rcl}
J_{\phi,\omega}\col{X} &=& \col{X^\ast},\\
\Delta^{\frac12}_{\phi,\omega}\ket{\Omega} &=&
\ket{\Phi},\\
\Delta^{\mathrm{i}t}_{\phi,\omega} &=& D^{\mathrm{i}t}_\phi \ot
\Pa{D^{-{\mathrm{i}t}}_\omega}^\t.
\end{array}
\right.
\end{eqnarray}

\section{Specific form of natural positive cone}

Let $\cH_d$ be a $d$-dimensional complex Hilbert space. Consider a
von Neumann algebra $\cM \equiv \lin{\cH_d}$. A faithful
representation of von Neumann algebra $\cM$ on a Hilbert space
$\cH\equiv\cH_d\ot\cH_d$ is defined by the following map:
$$
\pi: X\longmapsto \pi(X) = X\ot\I_d.
$$
$\ket{\Omega} \defeq \col{\I_d}$ is a separating and cyclic vector in
$\cH$ for von Neumann algebra $\pi(\cM) \equiv \lin{\cH_d}\ot\I_d$.
Thus von Neumann algebra $\cM$ have a standard representation
$(\pi(\cM),\cH,\ket{\Omega})$.

According to the definition of the natural positive cone $\cP$
associated with the pair $(\pi(\cM),\ket{\Omega})$ is the closure of the
set:
$$
\Set{\pi(M)j(\pi(M))\ket{\Omega}: M\in\cM},
$$
where $j: \pi(\cM)\longmapsto \pi(\cM)'$ is the anti-linear
$\ast$-isomorphism defined by
$$
j(\pi(M)) \defeq J\pi(M)J, \quad \forall M\in\cM.
$$
More concretely,
\begin{eqnarray*}
\pi(M)j(\pi(M))\ket{\Omega} &=& (M\ot \I_d)J(M\ot \I_d)J\col{\I_d}\\
&=& (M\ot \I_d)J(M\ot \I_d)\col{\I_d}\\
&=& (M\ot \I_d)J\col{M} = (M\ot \I_d)\col{M^\ast}\\
&=& \col{MM^\ast},
\end{eqnarray*}
which indicate that
$$
\cP = \Br{\Set{\col{MM^\ast}: M\in\cM}} = \Br{\col{\cM^+}} =
\col{\cM^+}.
$$
For any $\ket{\xi}\in\cP$, there exists an element $X\in\cM^+$ such that
$\ket{\xi} = \col{X}$, thus $J\col{X} = \col{X^\ast} = \col{X}$ since
$X=X^\ast$. Therefore $J\ket{\xi} = \ket{\xi}$. For any $\col{NN^\ast}\in\cP$
for some $N\in\cM$, we have
$$
\pi(M)j(\pi(M))\col{NN^\ast} = \col{MNN^\ast M^\ast} =
\col{(MN)(MN)^\ast}\in\cP.
$$

$\ket{\xi},\ket{\eta}$ are any given vectors in $\cP$. There exist two elements
$X,Y\in\cM^+$ such that $\ket{\xi} = \col{X}$ and $\ket{\eta}=\col{Y}$. Then
$$
\iinner{\xi}{\eta} = \Inner{\col{X}}{\col{Y}} = \Inner{X}{Y}_{\bH\bS}
= \Tr{XY}\geqslant0
$$
since $X,Y\geqslant0$. Thus $\cP$ is a self-dual cone. If
$Z\in\cM^+$ such that $\col{Z}\in\cP\cap(-\cP)$, then
$\col{Z}\in\cP$ and $\col{-Z}\in\cP$, which implies that
$-Z,Z\geqslant0$, i.e. $Z=0\Longleftrightarrow\col{Z}=0$. Therefore
$\cP\cap(-\cP)=\set{0}$.

If $\ket{\zeta}$ satisfies that $J\ket{\zeta}=\ket{\zeta}$, then there is an element
$T\in\cM$ such that $\ket{\zeta}=\col{T}$ and $\col{T}=J\col{T}$, which is
equivalent to the following formula:
$$
\col{T}=\col{T^\ast}\Longleftrightarrow T=T^\ast.
$$
Now by employing the Jordan decomposition of operators, we have
$$
T = T^+ - T^-,
$$
where $T^+,T^-\in\cM^+$ and $T^+ T^-=0$. This means that
$$
\ket{\zeta} = \col{T} = \col{T^+} - \col{T^-}
$$
and $\Inner{\col{T^+}}{\col{T^-}}=\Inner{T^+}{T^-}_{\bH\bS}=\Tr{T^+
T^-}=0$. Denote $\ket{\zeta_1} = \col{T^+}$ and $\ket{\zeta_2} = \col{T^-}$, then
$\ket{\zeta} = \ket{\zeta_1} - \ket{\zeta_2}$ with $\ket{\zeta_1}\bot\ket{\zeta_2}$.

For any $\ket{\varsigma}\in\cH$, there is an element $Y_{\ket{\varsigma}}\in\cM$
such that $\ket{\varsigma} = \col{Y_{\ket{\varsigma}}}$. Now since $Y_{\ket{\varsigma}}$ can
be represented by at most four positive element in
$H^+,H^-,K^+,K^-\in\cM^+$ as follows:
$$
Y_{\ket{\varsigma}} = (H^+ - H^-) + \mathrm{i}(K^+ - K^-),
$$
i.e.,
$$
\col{Y_{\ket{\varsigma}}} = \col{H^+} - \col{H^-} + \mathrm{i}\col{K^+} -
\col{K^-}.
$$
Setting
$\col{H^+}=\ket{\varsigma_1},\col{H^-}=\ket{\varsigma_2},\col{K^+}=\ket{\varsigma_3}$
and $\col{K^-}=\ket{\varsigma_4}$, we have
$$
\ket{\varsigma} = \ket{\varsigma_1} - \ket{\varsigma_2} + \mathrm{i}\ket{\varsigma_3} -
\mathrm{i}\ket{\varsigma_4}.
$$
Clearly, $\ket{\varsigma_1}, \ket{\varsigma_2}, \ket{\varsigma_3}, \ket{\varsigma_4}\in\cP$.
Finally, $\cH$ indeed is linearly spanned by $\cP$.

Since any normal positive form $\omega\in\cM_{\ast,+}$, it follows
that $\ket{\Omega} = \col{D^{\frac12}_\omega}$ is the vector representative
of $\ket{\omega}$ in $\cP$: $\omega(M) = \Innerm{\Omega}{\pi(M)}{\Omega}$.

Given any normal positive forms $\omega_{\ket{\xi}}$ and $\omega_{\ket{\eta}}$ for
$\ket{\xi},\ket{\eta}\in\cP$, thus we have $\ket{\xi} = \col{X}$ and $\ket{\eta} = \col{Y}$ for
$X,Y\in\cM^+$:
\begin{eqnarray}
\norm{\ket{\xi} - \ket{\eta}}^2 &=& \iinner{\xi-\eta}{\xi-\eta} =
\Inner{\col{X-Y}}{\col{X-Y}}_{\bH\bS} \nonumber\\
&=& \norm{X-Y}^2_{\bH\bS},\\
\norm{\ket{\xi} - \ket{\eta}}\norm{\ket{\xi} + \ket{\eta}} &=&
\norm{X-Y}_{\bH\bS}\norm{X+Y}_{\bH\bS},\\
\norm{\omega_{\ket{\xi}} - \omega_{\ket{\eta}}} &=& \norm{X^2-Y^2}_1.
\end{eqnarray}
By the result in Theorem~\ref{th:positive-cone},  it follows that
\begin{eqnarray}\label{eq-modular-ineq}
\norm{\ket{\xi} - \ket{\eta}}^2 \leqslant \norm{\omega_{\ket{\xi}} -
\omega_{\ket{\eta}}} \leqslant \norm{\ket{\xi} -
\ket{\eta}}\norm{\ket{\xi} + \ket{\eta}}.
\end{eqnarray}
Thus we arrived at the following inequality (a special case of
Powers-St\"{o}rmer's inequality):
\begin{thrm}[\cite{Watrous}]
It holds that
\begin{eqnarray}\label{eq-stormer}
\norm{X-Y}^2_{\bH\bS}\leqslant
\norm{X^2-Y^2}_1\leqslant\norm{X-Y}_{\bH\bS}\norm{X+Y}_{\bH\bS},
\end{eqnarray}
where $X,Y$ are Hermitian matrices.
\end{thrm}
In what follows, we can first show that Eq.~\eqref{eq-stormer} is
true, and then Eq.~\eqref{eq-modular-ineq} is a direct consequence
of Eq.~\eqref{eq-stormer}.
\begin{proof}
Since
$$
X^2 - Y^2 = \frac12\Br{(X-Y)(X+Y)+(X+Y)(X-Y)},
$$
it follows that
\begin{eqnarray*}
\norm{X^2-Y^2}_1 &\leqslant& \frac12\norm{(X-Y)(X+Y)}_1 +
\frac12\norm{(X+Y)(X-Y)}_1.
\end{eqnarray*}
By employing Schwarz inequality, we have
$$
\Set{\begin{array}{c}
         \norm{(X-Y)(X+Y)}_1 \\
         \norm{(X+Y)(X-Y)}_1
       \end{array}
}\leqslant\norm{X-Y}_{\bH\bS}\norm{X+Y}_{\bH\bS}.
$$
Thus
$$
\norm{X^2-Y^2}_1\leqslant\norm{X-Y}_{\bH\bS}\norm{X+Y}_{\bH\bS}.
$$
Next, we write the spectral decomposition of $X-Y$ as follows:
$$
X-Y = \sum_i\lambda_i \out{u_i}{u_i}.
$$
Then
$$
\abs{X-Y} = \sum_i\abs{\lambda_i}\out{u_i}{u_i},\quad
\Innerm{u_i}{X-Y}{u_i} = \lambda_i.
$$
Denote
$$
U \defeq \sum_i \sign(\lambda_i)\out{u_i}{u_i}.
$$
Thus $[U,X-Y]=0$ and $\abs{X-Y} = U(X-Y) = (X-Y)U$. Now by the
triangle inequality, we have
\begin{eqnarray}
\abs{\lambda_i} &=& \abs{\Innerm{u_i}{X-Y}{u_i}}=
\abs{\Innerm{u_i}{X}{u_i} - \Innerm{u_i}{Y}{u_i}} \nonumber\\
&\leqslant& \Innerm{u_i}{X}{u_i} + \Innerm{u_i}{Y}{u_i}\nonumber\\
&\leqslant&\Innerm{u_i}{X+Y}{u_i}.
\end{eqnarray}
Therefore
\begin{eqnarray*}
\norm{X^2 - Y^2}_1 &\geqslant& \abs{\Tr{\Br{X^2-Y^2} U}}\\
&=& \abs{\frac12\Tr{(X-Y)(X+Y)U} + \frac12\Tr{(X+Y)(X-Y)U}}\\
&=& \frac12\abs{\Tr{\abs{X-Y}(X+Y)} + \Tr{(X+Y)\abs{X-Y}}}\\
&=& \Tr{\abs{X-Y}(X+Y)} = \sum_i
\abs{\lambda_i}\Tr{\out{u_i}{u_i}(X+Y)}\\
&=& \sum_i \abs{\lambda_i}\Innerm{u_i}{X+Y}{u_i}\geqslant \sum_i
\abs{\lambda_i}^2=\norm{X-Y}^2_{\bH\bS}.
\end{eqnarray*}
The desired inequality is obtained.
\end{proof}

Powers-St\"{o}rmer's inequality asserts that for $s\in[0,1]$, the
following inequality
\begin{eqnarray}
2\Tr{A^s B^{1-s}} \geqslant \Tr{A+B-\abs{A-B}}
\end{eqnarray}
holds for any pair of positive matrices $A, B$. This is a key
inequality to prove the upper bound of Chernoff bound, in quantum
hypothesis testing theory \cite{Audenaert}. This inequality was
first proven by Audenaert, using an integral representation of the
function $t^s$. After that, Ozawa gave a much simpler proof for the
same inequality, using fact \cite{Jaksic} that $f(t)=t^s,
t\in[0,+\infty)$ is an operator monotone function for $s\in[0,1]$.

\begin{thrm}[Powers-St\"{o}rmer inequality \cite{Powers}] For positive compact operators
$A,B$, the following inequality is valid:
$$
\norm{\sqrt{A} - \sqrt{B}}^2_2 \leqslant \norm{A-B}_1.
$$
\end{thrm}

\begin{thrm}
Let $A,B$ be semi-definite positive matrices in $M_n(\complex)$.
Then
$$
2\Tr{B^sA^{1-s}}\geqslant \Tr{A+B-\abs{A-B}}
$$
holds for any $s\in[0,1]$.
\end{thrm}

\begin{proof}(Ozawa, unpublished) For $X$ self-adjoint, $X_{\pm}$
denotes its positive/negative part. Decomposing $A-B = (A-B)_+ -
(A-B)_-$, one gets
$$
\frac12\Tr{A+B-\abs{A-B}} = \Tr{A} - \Tr{(A-B)_+}.
$$
Now the original inequality is equivalent to
\begin{eqnarray}
\Tr{A} - \Tr{B^sA^{1-s}} \leqslant \Tr{(A-B)_+}.
\end{eqnarray}
Note that
$$
B+(A-B)_+\geqslant B\quad\text{and}\quad B+(A-B)_+ = A +
(A-B)_-\geqslant A.
$$
Since, for $s\in[0,1]$, the function $x\mapsto x^s$ is operator
monotone, i.e. $X\leqslant Y\Longrightarrow X^s \leqslant Y^s$ for
any positive matrices $X,Y$, we can write
\begin{eqnarray*}
\Tr{A} - \Tr{B^s A^{1-s}} &=& \Tr{(A^s - B^s)A^{1-s}}\\
&\leqslant& \Tr{((B+(A-B)_+)^s - B^s)A^{1-s}}\\
&\leqslant& \Tr{((B+(A-B)_+)^s - B^s)(B+(A-B)_+)^{1-s}}\\
&=& \Tr{B + (A-B)_+} - \Tr{B^s(B+(A-B)_+)^{1-s}}\\
&\leqslant& \Tr{B+(A-B)_+}.
\end{eqnarray*}
\end{proof}

\begin{thrm}[Ogata \cite{Ogata}]
Let $\phi_1,\phi_2$ are normal positive linear functionals on a von
Neumann algebra $\cM$ for which the vector representatives in the
natural positive cone $\cP$ are $\ket{\Phi_1}$ and $\ket{\Phi_2}$, respectively.
Then we have that, $\forall s\in[0,1]$,
\begin{eqnarray}
2\norm{\Delta^{\frac s2}_{\phi_2,\phi_1}\ket{\Phi_1}}^2\geqslant
\phi_1(\I) + \phi_2(\I) - \abs{\phi_1-\phi_2}(\I).
\end{eqnarray}
The equality holds if and only if
$$
\phi_2 = (\phi_2 - \phi_1)_+ + \psi\quad \text{and}\quad \phi_1 =
(\phi_2 - \phi_1)_- + \psi
$$
for some normal positive linear functional $\psi$ on $\cM$ whose
support is orthogonal to the support of $\abs{\phi_2-\phi_1}$.
\end{thrm}

\begin{thrm}[Hoa \cite{Hoa}]
Let $f$ be a $2n$-monotone function on $[0,+\infty)$ such that
$f((0,+\infty))\subseteq (0,+\infty)$. Then for any pair of positive
matrices $A,B\in M_n(\complex)$, we have:
\begin{eqnarray}
2\Tr{\sqrt{f(A)}g(B)\sqrt{f(A)}}\geqslant \Tr{A+B-\abs{A-B}},
\end{eqnarray}
where
$$
g(t) \defeq
\begin{cases}
\frac t{f(t)}, & t\in(0,+\infty),\\
0, & t = 0.
\end{cases}
$$
\end{thrm}

\begin{thrm}[Hoa \cite{Hoa}]
Let $\tau$ be a tracial functional on a $C^\ast$-algebra $\cA$, $f$
be a strictly positive, operator monotone function on $[0,+\infty)$.
Then for any pair of positive elements $A,B\in\cA$:
\begin{eqnarray}
2\tau\Pa{\sqrt{f(A)}g(B)\sqrt{f(A)}} \geqslant \tau(A+B-\abs{A-B}),
g(t) \defeq t/f(t).
\end{eqnarray}
\end{thrm}

\begin{thrm}[Phillips \cite{Phillips}] Let
$A\geqslant B\geqslant0$ and $t\geqslant1$. Then
$$
\norm{A^{1/t}-B^{1/t}}^t_t \leqslant \norm{A-B}_1.
$$
\end{thrm}

Let $\cM$ be a general von Neumann algebra with a faithful normal
semi-finite weight $\varphi$. Denote by $\cN$ the crossed product
$\cM\rtimes_{\sigma^\varphi}\real$ which admits the canonical
faithful normal semi-finite trace $\tau$ and the dual action
$\theta_s(s\in\real)$, satisfying $\tau\circ\theta_s =
e^{-s}\tau(s\in\real)$. For $p\in(0,\infty]$, the \emph{Haagerup
$L^p$-space} $L^p(\cM) = L^p(\cM;\varphi)$ is defined by
$$
L^p(\cM) \defeq \Set{X\in\widetilde{\cN}:
\theta_s(X)=e^{-s/p}X,s\in\real}.
$$
Here $\cM=L^\infty(\cM)$. For each $\psi\in\cM^+_\ast$, a unique
$D_\psi\in\widetilde{\cN}^+$ is given by
$\widetilde{\psi}=\tau(D_\psi\cdot)$, where $\widetilde{\psi}$ is
the dual weight of $\psi$. The mapping $\psi\mapsto D_\psi$ is
extended to a linear bijection from $\cM_\ast$ onto $L^1(\cM)$, and
so the linear functional $\trace$ on $L^1(\cM)$ is defined by
$\Tr{D_\psi}=\psi(\I)(\psi\in\cM_\ast)$.

For $p\in(0,\infty)$, the Haagerup (quasi-)norm $\norm{X}_p$ of
$X\in L^p(\cM)$ is defined by $\norm{X}_p=\Tr{\abs{X}^p}^{1/p}$.
When $p\in[1,\infty)$, $L^p(\cM)$ is a Banach space with the norm
$\norm{\cdot}_p$, and its dual Banach space is $L^q(\cM)$, where
$\frac1p+\frac1q=1$ by the following duality:
$$
\Inner{X}{Y} = \Tr{XY} (=\Tr{YX}),\quad X\in L^p(\cM),Y\in L^q(\cM).
$$
In particular, $\cM_\ast\cong L^1(\cM)$ by the isometry $\psi\mapsto
D_\psi$.
\begin{thrm}[Hiai \cite{Hiai}]
Let $L^p(\cM)$ be the Haagerup $L^p$-space for some von Neumann
algebra. For $A,B\in L^p(\cM)^+$, we have
$$
\norm{A^t - B^t}_{p/t}\leqslant \norm{A-B}^t_p,
$$
where $t\in(0,1)$ and $p\in[t,\infty]$.
\end{thrm}

\section{Effros' approach---applications of Araki relative modular operator}

\subsection{The classical and matrix notions of perspectives}

Given a convex function $f$ defined on a convex set
$\cC\subseteq\real^n$, the \emph{perspective} $g$ is defined on the
subset
$$
L \defeq \Set{(x,t): t>0\ \text{and}\ x/t\in \cC}
$$
by
$$
g(x,t) \defeq f(x/t)t.
$$
It is a simple exercise to verify that $g(x, t)$ is a jointly convex
function in the sense that, if $\lambda\in[0,1]$
$$
g(\lambda x_1 + (1-\lambda)x_2, \lambda t_1 + (1-\lambda)t_2)
\leqslant \lambda g(x_1,t_1) + (1-\lambda)g(x_2,t_2).
$$
An elementary but important example is provided by the continuous
convex function $f(x) = x\log x$, with $f(0)=0$ defined on
$[0,+\infty)\subset\real$. It follows that the perspective function
$$
g(x,t) = t\frac xt\log\frac xt = x\log x - x\log t
$$
is jointly convex. Letting $p=(p_i)$ and $q=(q_i)$ be finite
probability measures with $p_i>0$ and $q_i>0$, the convexity of $f$
implies that the classical entropy
$$
\rH(p) = - \sum_i p_i\log p_i
$$
is concave, and the convexity of $g$ implies that the relative
entropy
$$
(q,p)\mapsto \rH(q||p) = \sum_i p_i\log p_i - p_i\log q_i
$$
is jointly convex on pairs of probability measures.

We recall that if $f: I = [a, b] \to R$ is continuous, and $T$ is an
$n\times n$ self-adjoint matrix with spectrum in $[a, b]$, then we
can define $f_n(T)$ by spectral theory (or by using a basis in which
$T$ is diagonal). $f$ is said to be \emph{matrix convex} if for each
$n \in \natural$, the corresponding function $f_n$ is convex on the
self-adjoint $n\times n$ matrices with spectrum in $[a, b]$.
Throughout the rest of the article we only consider $n\times n$
matrices, and we usually omit the subscript $n$. The following is
the affine version of the Hansen-Pedersen-Jensen inequality

\begin{thrm}
If $f$ is matrix convex, and $A$ and $B$ satisfy $A^\dagger A +
B^\dagger B = \I_n$, then
$$
f(A^\dagger T_1 A + B^\dagger T_2 B)\leqslant A^\dagger f(T_1) A +
B^\dagger f(T_2) B.
$$
\end{thrm}

\begin{thrm}[Effros \cite{Effros}]\label{th:Effros}
Suppose that $f(x)$ is operator convex. When restricted to positive
commuting matrices $\mathbb{L},\mathbb{R}$, i.e.
$[\mathbb{L},\mathbb{R}]=0$, the "perspective function"
\begin{eqnarray}
(\mathbb{L},\mathbb{R})\mapsto g(\mathbb{L},\mathbb{R}) =
f(\mathbb{L}/\mathbb{R}) \mathbb{R}
\end{eqnarray}
is jointly convex in the sense that if
$$
\mathbb{L} = \lambda \mathbb{L}_1 +
(1-\lambda)\mathbb{L}_2\quad\text{and}\quad \mathbb{R} = \lambda
\mathbb{R}_1 + (1-\lambda)\mathbb{R}_2
$$
with $[\mathbb{L}_i,\mathbb{R}_i]=0$ $(i=1,2)$, $\lambda\in[0,1]$,
\begin{eqnarray}
g(\mathbb{L},\mathbb{R}) \leqslant \lambda
g(\mathbb{L}_1,\mathbb{R}_1) + (1-\lambda) g(\mathbb{L}_2,
\mathbb{R}_2).
\end{eqnarray}
\end{thrm}

\begin{proof}
The matrices $A =
(\lambda\mathbb{R}_1)^{\frac12}\mathbb{R}^{-\frac12}$ and $B =
((1-\lambda)\mathbb{R}_2)^{\frac12}\mathbb{R}^{-\frac12}$ satisfy
$A^\dagger A + B^\dagger B = \I$. From the above Theorem, we have
\begin{eqnarray*}
g(\mathbb{L},\mathbb{R}) &=& \mathbb{R}f(\mathbb{L}/\mathbb{R}) = \mathbb{R}^{\frac12}f(\mathbb{R}^{-\frac12}\mathbb{L}\mathbb{R}^{-\frac12})\mathbb{R}^{\frac12}\\
&=& \mathbb{R}^{\frac12}f(A^\dagger(\mathbb{L}_1/\mathbb{R}_1)A + B^\dagger(\mathbb{L}_2/\mathbb{R}_2)B)\mathbb{R}^{\frac12}\\
&\leqslant& \mathbb{R}^{\frac12}\Pa{A^\dagger f(\mathbb{L}_1/\mathbb{R}_1)A + B^\dagger f(\mathbb{L}_2/\mathbb{R}_2)B}\mathbb{R}^{\frac12}\\
&=& (\lambda \mathbb{R}_1)^{\frac12}
f(\mathbb{L}_1/\mathbb{R}_1)(\lambda \mathbb{R}_1)^{\frac12} +
((1-\lambda) \mathbb{R}_2)^{\frac12} f(\mathbb{L}_2/\mathbb{R}_2)((1-\lambda) \mathbb{R}_2)^{\frac12}\\
&=& \lambda g(\mathbb{L}_1,\mathbb{R}_1) + (1-\lambda)
g(\mathbb{L}_2,\mathbb{R}_2).
\end{eqnarray*}
\end{proof}

\begin{cor}
The relative entropy function
$$
(\rho,\sigma)\mapsto \rS(\rho||\sigma) = \Tr{\rho\log\rho
-\rho\log\sigma}
$$
is jointly convex on the strictly positive $n\times n$ density
matrices $\rho,\sigma$.
\end{cor}

\begin{proof}
The function $f(x) = x\log x$ is operator convex and thus
$$
\Inner{\I}{g(\mathbb{L}_\rho,\mathbb{R}_\sigma)(\I)} =
\rS(\rho||\sigma)
$$
is jointly convex.
\end{proof}

\begin{cor}
If $s\in(0,1)$, then the function
$$
F(A,B) = \Tr{A^s K^\dagger B^{1-s} K}
$$
is jointly concave on the strictly positive $n\times n$ matrices
$A,B$.
\end{cor}

\begin{proof}
$f(t) = -t^s$ is operator convex, $- \Tr{A^s K^\dagger B^{1-s} K} =
\Inner{K^\dagger}{g(\mathbb{L}_A,\mathbb{R}_B)(K^\dagger)}$ is
jointly convex.
\end{proof}

\subsection{Mar\'{e}chal's perspectives}

We assume that the functions $f$ and $g$ are defined on an interval
$I\subseteq \real$ and that $0\in I$.

\begin{thrm}\label{th:Marechal}
If $f$ is matrix convex, $f(0)\leqslant0$, and $A$ and $B$ are
matrices with $A^\dagger A + B^\dagger B\leqslant\I_n$, then
$$
f(A^\dagger T_1 A + B^\dagger T_2 B) \leqslant A^\dagger f(T_1) A +
B^\dagger f(T_2) B.
$$
\end{thrm}
Given continuous functions $f$ and $h$, and commuting positive
matrices $\mathbb{L}$ and $\mathbb{R}$, we define
$$
(f\Delta h) (\mathbb{L},\mathbb{R}) \defeq
f(\mathbb{L}/h(\mathbb{R}))h(R)
$$
\begin{thrm}
Suppose that $f$ is matrix convex, $f(0)\leqslant0$, and that $h$ is
matrix concave with $h>0$. Then
$$
(\mathbb{L},\mathbb{R})\mapsto (f\Delta h) (\mathbb{L},\mathbb{R})
$$
is jointly convex on positive commuting matrices
$\mathbb{L},\mathbb{R}$ in the sense of Theorem~\ref{th:Effros}.
\end{thrm}

\begin{proof}
Let us suppose that $\mathbb{L} = \lambda \mathbb{L}_1 +
(1-\lambda)\mathbb{L}_2$ and $\mathbb{R} = \lambda \mathbb{R}_1 +
(1-\lambda)\mathbb{R}_2$ where $\Br{\mathbb{L}_j,\mathbb{R}_j} =0$.
Then
$$
\lambda h(\mathbb{R}_1) + (1-\lambda)h(\mathbb{R}_2) \leqslant
h(\mathbb{R}),
$$
hence
$$
A = (\lambda h(\mathbb{R}_1))^{\frac12}
(h(\mathbb{R}))^{-\frac12},\quad B = ((1-\lambda)
h(\mathbb{R}_2))^{\frac12} (h(\mathbb{R}))^{-\frac12}
$$
satisfy
$$
A^\dagger A + B^\dagger B\leqslant \I.
$$
It follows from Theorem~\ref{th:Marechal} that
\begin{eqnarray*}
(f\Delta h) (\mathbb{L},\mathbb{R}) &=& (h(\mathbb{R}))^{\frac12}f((h(\mathbb{R}))^{-\frac12}\mathbb{L}(h(\mathbb{R}))^{-\frac12})(h(\mathbb{R}))^{\frac12}\\
&=& (h(\mathbb{R}))^{\frac12}f(A^\dagger(\mathbb{L}_1/\mathbb{R}_1)A + B^\dagger(\mathbb{L}_2/\mathbb{R}_2)B)(h(\mathbb{R}))^{\frac12}\\
&\leqslant& (h(\mathbb{R}))^{\frac12}A^\dagger f(\mathbb{L}_1/\mathbb{R}_1)A (h(\mathbb{R}))^{\frac12} + (h(\mathbb{R}))^{\frac12}B^\dagger f(\mathbb{L}_2/\mathbb{R}_2)B (h(\mathbb{R}))^{\frac12}\\
&=& \lambda (f\Delta h)(\mathbb{L}_1,\mathbb{R}_1) +
(1-\lambda)(f\Delta h)(\mathbb{L}_2,\mathbb{R}_2).
\end{eqnarray*}
\end{proof}

\begin{cor}
Suppose that $0<p,q$ and that $p+q\leqslant1$. Then the function
$$
(A,B)\mapsto \Tr{A^q X^\dagger B^p X}
$$
is jointly concave on the positive $n\times n$ matrices.
\end{cor}

\begin{proof}
Since $p+q\leqslant 1$, $p+q$ is a convex combination of $q$ and 1,
i.e. we may choose $t\in[0,1]$ with $p+q = (1-t)q + t1$. If we let
$q=s$, then
$$
p = - tq + t = (1-q)t = (1-s)t.
$$
Thus, it suffices to show that if $s,t\in[0,1]$, then
$$
(A,B)\mapsto - \Tr{A^s X^\dagger B^{(1-s)t}X}
$$
is jointly convex. The functions $f(x) = -x^s$ and $h(y)=y^t$ are
operator convex and concave, respectively, and
$$
(f\Delta h)(\mathbb{L}_A,\mathbb{R}_B) =
h(\mathbb{R}_A)f(\mathbb{L}_A/h(\mathbb{R}_B)) = -
\mathbb{L}_A^s\mathbb{R}_B^{(1-s)t}.
$$
Therefore
$$
- \Tr{A^s X^\dagger B^{(1-s)t}X} = \Inner{X^\dagger}{(f\Delta
h)(\mathbb{L}_A,\mathbb{R}_B)(X^\dagger)}
$$
is jointly convex.
\end{proof}

\begin{lem}\label{lem:positivityforblock}
Let $A,B\in\pd{\cH}$ and $X\in\lin{\cH}$. Then
$\Br{\begin{array}{cc}
                                                      A & X \\
                                                      X^\dagger & B
                                                    \end{array}
}$ is positive semi-definite if and only if $B\geqslant X^\ast
A^{-1}X$.
\end{lem}

\begin{lem}\label{lem:operatorinequality}
Let $A,B,A_i,B_i\in\pd{\cH}(i=1,2)$ be such that
$[A,B]=[A_i,B_i]=0(i=1,2)$, and
$$
A\geqslant \lambda A_1 + (1-\lambda)A_2,\quad B\geqslant \lambda B_1
+ (1-\lambda)B_2,\quad \forall \lambda\in[0,1].
$$ Then
\begin{eqnarray}\label{eq:superadditivity}
A^s B^t \geqslant \lambda A^s_1 B^t_1 + (1-\lambda) A^s_2 B^t_2
\end{eqnarray}
for all $0\leqslant s,t; s+t\leqslant1$.
\end{lem}

\begin{proof}
Let $E$ be the set of all $(s,t)\in[0,1]\times[0,1]$ for which the
inequality Eq.~(\ref{eq:superadditivity}) holds. We first show that
$(\frac12,\frac12)\in E$. From Lemma~\ref{lem:positivityforblock},
it follows that
$$
\left[
  \begin{array}{cc}
    A_1 & \sqrt{A_1B_1} \\
    \sqrt{A_1B_1} & B_1 \\
  \end{array}
\right]\quad \text{and}\quad \left[
  \begin{array}{cc}
    A_2 & \sqrt{A_2B_2} \\
    \sqrt{A_2B_2} & B_2 \\
  \end{array}
\right]
$$
are positive semi-definite. Consequently,
$$
\left[
  \begin{array}{cc}
    \lambda A_1 + (1-\lambda)A_2 & \lambda\sqrt{A_1B_1} + (1-\lambda)\sqrt{A_2B_2} \\
    \lambda\sqrt{A_1B_1} + (1-\lambda)\sqrt{A_2B_2} & \lambda B_1 + (1-\lambda)B_2 \\
  \end{array}
\right]
$$
is positive semi-definite. Using the facts that $A\geqslant\lambda
A_1 + (1-\lambda)A_2$ and $B\geqslant\lambda B_1 + (1-\lambda)B_2$,
we see that
$$
\left[
  \begin{array}{cc}
    A & \lambda\sqrt{A_1B_1} + (1-\lambda)\sqrt{A_2B_2} \\
    \lambda\sqrt{A_1B_1} + (1-\lambda)\sqrt{A_2B_2} & B \\
  \end{array}
\right]
$$
is positive semi-definite and hence
$$
\left[
  \begin{array}{cc}
    \I & A^{-\frac12}\Br{\lambda\sqrt{A_1B_1} + (1-\lambda)\sqrt{A_2B_2}}A^{-\frac12} \\
    A^{-\frac12}\Br{\lambda\sqrt{A_1B_1} + (1-\lambda)\sqrt{A_2B_2}}A^{-\frac12} & A^{-\frac12}BA^{-\frac12} \\
  \end{array}
\right]
$$
is positive semi-definite. Thus, again by
Lemma~\ref{lem:positivityforblock}, we have
$$
A^{-1}B = A^{-\frac12}BA^{-\frac12} \geqslant
\Br{A^{-\frac12}\Br{\lambda\sqrt{A_1B_1} +
(1-\lambda)\sqrt{A_2B_2}}A^{-\frac12}}^2.
$$
Using the fact that the function $g(x)=\sqrt{x}$ is operator
monotone on $[0,+\infty)$, we get
$$
\sqrt{A^{-1}B} \geqslant A^{-\frac12}\Br{\lambda\sqrt{A_1B_1} +
(1-\lambda)\sqrt{A_2B_2}}A^{-\frac12},
$$
which implies
\begin{eqnarray}\label{eq:squarerootsuperadd}
\sqrt{AB} \geqslant \lambda\sqrt{A_1B_1} + (1-\lambda)\sqrt{A_2B_2}.
\end{eqnarray}
This proves that $(\frac12,\frac12)\in E$. Clearly,
$(0,0),(0,1),(1,0)$ are in $E$ and $E$ is closed. If
$(s_1,t_1),(s_2,t_2)\in E$, then it follows as the proof of
Eq.~(\ref{eq:squarerootsuperadd}) that
$(\frac{s_1+s_2}2,\frac{t_1+t_2}2)\in E$, and so $E$ is convex. This
proves the lemma.
\end{proof}

The following theorem is known as Lieb's concavity theorem.

\begin{thrm}[Lieb's concavity theorem]
Let $X\in\lin{\cH}$ and $s,t\geqslant0$ be such that
$s+t\leqslant1$. Then the map
$$
f(A,B) = \Tr{X^\ast A^s X B^t}
$$
is jointly concave on $\pd{\cH}\times\pd{\cH}$.
\end{thrm}

\begin{proof}
Let $A_i,B_i\in\pd{\cH}(i=1,2)$ and $\lambda\in[0,1]$. Let
$\mathbb{L}_{A_i}, \mathbb{L}_A$ be the left multiplication
operators on the space $\lin{\cH}$ induced by $A_i,A=\lambda
A_1+(1-\lambda)A_2$, where $i=1,2$, respectively; $\mathbb{R}_{B_i},
\mathbb{R}_B$ be the right multiplication operators on the space
$\lin{\cH}$ induced by $B_i,B=\lambda B_1+(1-\lambda)B_2$, where
$i=1,2$, respectively. The
$\mathbb{L}_{A_i},\mathbb{L}_A,\mathbb{R}_{B_i},\mathbb{R}_B$ are
positive operators on $\lin{\cH}$. Moreover, $\mathbb{L}_{A_i}$
commutes with $\mathbb{R}_{B_i}$, $\mathbb{L}_A$ commutes with
$\mathbb{R}_B$. Also, we have
$$
\mathbb{L}_A = \lambda\mathbb{L}_{A_1} +
(1-\lambda)\mathbb{L}_{A_2}\quad \text{and}\quad \mathbb{R}_B =
\lambda\mathbb{R}_{B_1} + (1-\lambda)\mathbb{R}_{B_2}.
$$
Therefore, by Lemma~\ref{lem:operatorinequality},
$$
\mathbb{L}_A^s\mathbb{R}_B^t\geqslant \lambda
\mathbb{L}_{A_1}^s\mathbb{R}_{B_1}^t +
(1-\lambda)\mathbb{L}_{A_2}^s\mathbb{R}_{B_2}^t
$$
for $0\leqslant s,t;s+t\leqslant1$. Thus, for every $X\in\lin{\cH}$,
\begin{eqnarray*}
\Inner{X}{\Pa{\mathbb{L}_A^s\mathbb{R}_B^t} (X)}&\geqslant&
\Inner{X}{\Pa{\lambda \mathbb{L}_{A_1}^s\mathbb{R}_{B_1}^t +
(1-\lambda)\mathbb{L}_{A_2}^s\mathbb{R}_{B_2}^t} (X)}\\
&=& \lambda\Inner{X}{\Pa{\mathbb{L}_{A_1}^s\mathbb{R}_{B_1}^t}(X)} +
(1-\lambda) \Inner{X}
{\Pa{\mathbb{L}_{A_2}^s\mathbb{R}_{B_2}^t}(X)}.
\end{eqnarray*}
That is,
$$
\Inner{X}{(\lambda A_1 + (1-\lambda)A_2)^s X(\lambda B_1 +
(1-\lambda)B_2)^t}\geqslant \lambda \Inner{X}{A^s_1 XB^t_1} +
(1-\lambda)\Inner{X}{A^s_2 XB^t_2}.
$$
This completes the proof.
\end{proof}

\begin{remark}
There are two key elements in this proof. One is the replacement of
the noncommuting matrices $A_i$ and $B_i$ by left and right
multiplication operators $\cA_i$ and $\cB_i$, respectively, which
act on matrices and commute. This idea is implicit in proofs based
on Araki's relative modular operator.
\end{remark}

\subsection{Operator extension of strong subadditivity of entropy}

Following Effros, we choose $\mathbb{L}_\rho$ and
$\mathbb{R}_\sigma$ to be superoperators that multiplies matrix from
the left or right. For $X \in \lin{\cH}$, $\mathbb{L}_\rho$ and
$\mathbb{R}_\sigma$ are defined as follows.
\begin{eqnarray}
\mathbb{L}_\rho X = \rho X\quad\text{and}\quad \mathbb{R}_\sigma X =
X \sigma.
\end{eqnarray}
Note in particular, that $\mathbb{L}_\rho$ and $\mathbb{R}_\sigma$
commute with each other. One can also show the following relations.
\begin{eqnarray}
\log (\mathbb{L}_\rho) X = \log (\rho) X\quad\text{and}\quad\log
(\mathbb{R}_\sigma) X = X\log (\sigma).
\end{eqnarray}

Denoting $\widehat{H}_A = - \log (\rho_A)\ot \I_{A^c}$, following
statement follows from Effros' result.
\begin{thrm}[Kim \cite{Kim}]
Let $\rho_{ABC}\in\density{\cH_A\ot\cH_B\ot\cH_C}$. Denote $\widehat
H_{X} = \log(\rho_X) \ot\I_{X^c}$, where $X\in\Set{AB,BC,B,ABC}$.
\begin{eqnarray}
\Ptr{AB}{\rho_{ABC}\Pa{\widehat{H}_{AB} + \widehat{H}_{BC} - \widehat{H}_B - \widehat{H}_{ABC}}} \geqslant 0,\\
\Ptr{BC}{\rho_{ABC}\Pa{\widehat{H}_{AB} + \widehat{H}_{BC} -
\widehat{H}_B - \widehat{H}_{ABC}}} \geqslant 0.
\end{eqnarray}
\end{thrm}

\begin{proof}
Let $f(x) = x\log x$. Since $f(x)$ is operator convex,
$$
g(\mathbb{L}_\rho,\mathbb{R}_\sigma)= \mathbb{L}_\rho \log
(\mathbb{L}_\rho) - \mathbb{L}_\rho \log (\mathbb{R}_\sigma)
$$
is jointly convex in $\mathbb{L}_\rho$ and $\mathbb{R}_\sigma$.
Therefore,
\begin{eqnarray}
\Inner{K}{g(\mathbb{L}_\rho,\mathbb{R}_\sigma)(K)} =
\Tr{\rho\log(\rho)KK^\dagger - \rho K\log(\sigma) K^\dagger}
\end{eqnarray}
is jointly convex in $\mathbb{L}_\rho$ and $\mathbb{R}_\sigma$ for
all $K \in \lin{\cH}$. Choose
$$
\rho = \rho_{ABC},\quad \sigma = \rho_{AB} \ot \I_C/d_C,\quad K =
\I_{AB} \ot P_C,
$$
where $P_C$ is a projector acting on $\cH_C$ and $d_C$ is dimension
of $\cH_C$. Note
\begin{eqnarray}
\I_A/d_A \ot \rho_{BC} = \frac{1}{d_A^2} \sum_{\mu=1}^{d_A^2}
U_{A,\mu} \rho_{ABC} U_{A,\mu}^{\dagger}
\end{eqnarray}
for some unitaries $\set{U_{A,\mu}}$. Using joint convexity, we see
that
\begin{eqnarray*}
&&\Tr{(\I_A/d_A \ot \rho_{BC})\Br{\log\Pa{\I_A/d_A \ot \rho_{BC}} - \log\Pa{\I_A/d_A \ot \rho_B \ot \I_C/d_C}}P_C}\\
&&\leqslant \frac{1}{d_A^2} \sum_{\mu=1}^{d_A^2}\Tr{(U_{A,\mu} \rho_{ABC} U_{A,\mu}^{\dagger})\Br{\log\Pa{U_{A,\mu} \rho_{ABC} U_{A,\mu}^{\dagger}} - \log\Pa{U_{A,\mu} \rho_{AB} U_{A,\mu}^{\dagger} \ot \I_C/d_C}}P_C}\\
&&= \Tr{\rho_{ABC}(\log(\rho_{ABC}) - \log(\rho_{AB} \ot
\I_C/d_C))P_C}.
\end{eqnarray*}
Now denote
\begin{eqnarray*}
\mathrm{L.H.S.} &\defeq& \Tr{(\I_A/d_A \ot \rho_{BC})\Br{\log\Pa{\I_A/d_A \ot \rho_{BC}} - \log\Pa{\I_A/d_A \ot \rho_B \ot \I_C/d_C}}P_C},\\
\mathrm{R.H.S.} &\defeq& \Tr{\rho_{ABC}(\log(\rho_{ABC}) -
\log(\rho_{AB} \ot \I_C/d_C))P_C}.
\end{eqnarray*}
Then
\begin{eqnarray*}
\mathrm{L.H.S.} &=& \Tr{\rho_{BC}\Pa{\widehat H_B - \widehat H_{BC}}P_C} + \log (d_C)\Tr{\rho_C P_C}\\
&=& \Tr{\rho_{ABC}\Pa{\widehat H_B - \widehat H_{BC}}P_C} + \log (d_C)\Tr{\rho_C P_C},\\
\mathrm{R.H.S.} &=& \Tr{\rho_{ABC}\Pa{\widehat H_{AB} - \widehat
H_{ABC}}P_C} + \log (d_C)\Tr{\rho_C P_C}.
\end{eqnarray*}
Since
$$
\Tr{\rho_{ABC}\Pa{\widehat{H}_{AB} + \widehat{H}_{BC} -
\widehat{H}_B - \widehat{H}_{ABC}}P_C} = \mathrm{R.H.S.} -
\mathrm{L.H.S.}\geqslant0
$$
holds for an arbitrary projector $P_C$, it follows that
$$
\Ptr{AB}{\rho_{ABC}\Pa{\widehat{H}_{AB} + \widehat{H}_{BC} -
\widehat{H}_B - \widehat{H}_{ABC}}}\geqslant 0.
$$
That is, $\Ptr{AB}{\rho_{ABC}\Pa{\widehat{H}_{AB} + \widehat{H}_{BC}
- \widehat{H}_B - \widehat{H}_{ABC}}}$ is a positive semi-definite
operator acting on $\cH_C$. Similarly, we have that
$$
\Ptr{BC}{\rho_{ABC}\Pa{\widehat{H}_{AB} + \widehat{H}_{BC} -
\widehat{H}_B - \widehat{H}_{ABC}}}\geqslant0.
$$
\end{proof}
One may wish to find a similar inequality when partial trace is
restricted to $A$ or $B$. In both cases, the resulting operators are
not even hermitian.

\begin{thrm}[Ruskai \cite{Ruskai}]
Let $\rho_{ABC}\in\density{\cH_A\ot\cH_B\ot\cH_C}$. Denote $\widehat
H_{X} = \log(\rho_X) \ot\I_{X^c}$, where $X\in\Set{AB,BC,B,ABC}$.
\begin{eqnarray}
\Ptr{AB}{\Pa{\widehat{H}_{AB} + \widehat{H}_{BC} - \widehat{H}_B - \widehat{H}_{ABC}}\rho_{ABC}} \geqslant 0,\\
\Ptr{BC}{\Pa{\widehat{H}_{AB} + \widehat{H}_{BC} - \widehat{H}_B -
\widehat{H}_{ABC}}\rho_{ABC}} \geqslant 0,\\
\Ptr{AB}{\rho_{AB}\Pa{\widehat{H}_{AB} + \widehat{H}_{BC} -
\widehat{H}_B - \widehat{H}_{ABC}}} \leqslant 0
\end{eqnarray}
\end{thrm}

\begin{cor}[Kim \cite{Kim}]
Let $\rho_{AB}\in\density{\cH_A\ot\cH_B}$. The we have:
\begin{eqnarray}
\Ptr{A}{\rho_{AB}\Pa{\widehat{H}_A + \widehat{H}_B -
\widehat{H}_{AB}}} \geqslant 0,\\
\Ptr{B}{\rho_{AB}\Pa{\widehat{H}_A + \widehat{H}_B -
\widehat{H}_{AB}}} \geqslant 0.
\end{eqnarray}
\end{cor}

\begin{cor}[Ruskai \cite{Ruskai}]
Let $\rho_{AB}\in\density{\cH_A\ot\cH_B}$. The we have:
\begin{eqnarray}
\Ptr{A}{\Pa{\widehat{H}_A + \widehat{H}_B -
\widehat{H}_{AB}}\rho_{AB}} \geqslant 0,\\
\Ptr{B}{\Pa{\widehat{H}_A + \widehat{H}_B -
\widehat{H}_{AB}}\rho_{AB}} \geqslant 0.
\end{eqnarray}
\end{cor}

\begin{exam}
Let $\rho$ be a state and $\mathbb{K}_\rho(X) \defeq \int^1_0 \rho^t
X\rho^{1-t}dt$ defined for Hermite matrices. Recall that if
$\rho=\sum_i\lambda_i\out{\lambda_i}{\lambda_i}$, then
$\rho^{-1}=\sum_i\frac1{\lambda_i}\out{\lambda_i}{\lambda_i}$. For a
super-operator, the spectral projection is
$\mathbb{L}_{\out{\lambda_i}{\lambda_i}}\mathbb{R}_{\out{\lambda_j}{\lambda_j}}$
for which its action is given by
\begin{eqnarray*}
\mathbb{L}_{\out{\lambda_i}{\lambda_i}}\mathbb{R}_{\out{\lambda_j}{\lambda_j}}\mathbb{K}_\rho(X)
&=& \Innerm{\lambda_i}{\int^1_0 \rho^t X\rho^{1-t}dt}{\lambda_j}
\out{\lambda_i}{\lambda_j} \\
&=& \int^1_0
\lambda^t_i\lambda^{1-t}_jdt\Innerm{\lambda_i}{X}{\lambda_j}\out{\lambda_i}{\lambda_j}\\
&=& \frac{\lambda_i - \lambda_j}{\ln\lambda_i - \ln\lambda_j}
\mathbb{L}_{\out{\lambda_i}{\lambda_i}}\mathbb{R}_{\out{\lambda_j}{\lambda_j}}(X),
\end{eqnarray*}
that is
$$
\mathbb{L}_{\out{\lambda_i}{\lambda_i}}\mathbb{R}_{\out{\lambda_j}{\lambda_j}}\mathbb{K}_\rho
=\mathbb{K}_\rho
\mathbb{L}_{\out{\lambda_i}{\lambda_i}}\mathbb{R}_{\out{\lambda_j}{\lambda_j}}
= \frac{\lambda_i - \lambda_j}{\ln\lambda_i - \ln\lambda_j}
\mathbb{L}_{\out{\lambda_i}{\lambda_i}}\mathbb{R}_{\out{\lambda_j}{\lambda_j}}.
$$
This gives that
$$
\mathbb{L}_{\out{\lambda_i}{\lambda_i}}\mathbb{R}_{\out{\lambda_j}{\lambda_j}}\mathbb{K}^{-1}_\rho
=\mathbb{K}^{-1}_\rho\mathbb{L}_{\out{\lambda_i}{\lambda_i}}\mathbb{R}_{\out{\lambda_j}{\lambda_j}}
= \frac{\ln\lambda_i - \ln\lambda_j}{\lambda_i - \lambda_j}
\mathbb{L}_{\out{\lambda_i}{\lambda_i}}\mathbb{R}_{\out{\lambda_j}{\lambda_j}}.
$$
Using the integral representation of $\ln x$:
$$
\ln x = \int^\infty_0 \Pa{\frac1{1+t} - \frac1{x+t}}dt,
$$
it follows that
\begin{eqnarray*}
\frac{\ln\lambda_i - \ln\lambda_j}{\lambda_i - \lambda_j} &=&
\frac1{\lambda_i - \lambda_j}\Br{\int^\infty_0 \Pa{\frac1{1+t} -
\frac1{\lambda_i+t}}dt - \int^\infty_0 \Pa{\frac1{1+t} -
\frac1{\lambda_j+t}}dt}\\
&=& \frac1{\lambda_i -
\lambda_j}\Br{\int^\infty_0\Pa{\frac1{\lambda_j+t} -
\frac1{\lambda_i+t}}dt}\\
&=&\int^\infty_0\frac1{(\lambda_i+t)(\lambda_j+t)}dt.
\end{eqnarray*}
Thus
$$
\mathbb{L}_{\out{\lambda_i}{\lambda_i}}\mathbb{R}_{\out{\lambda_j}{\lambda_j}}\mathbb{K}^{-1}_\rho
=
\int^\infty_0\frac1{(\lambda_i+t)(\lambda_j+t)}dt\mathbb{L}_{\out{\lambda_i}{\lambda_i}}
\mathbb{R}_{\out{\lambda_j}{\lambda_j}}.
$$
Furthermore,
\begin{eqnarray*}
\mathbb{L}_{\out{\lambda_i}{\lambda_i}}\mathbb{R}_{\out{\lambda_j}{\lambda_j}}\mathbb{K}^{-1}_\rho(X)
&=&
\int^\infty_0\frac1{(\lambda_i+t)(\lambda_j+t)}dt\mathbb{L}_{\out{\lambda_i}{\lambda_i}}
\mathbb{R}_{\out{\lambda_j}{\lambda_j}}(X)\\
&=&
\int^\infty_0\frac1{(\lambda_i+t)(\lambda_j+t)}dt\out{\lambda_i}{\lambda_i}X\out{\lambda_j}{\lambda_j}\\
&=&\int^\infty_0(\lambda_i+t)^{-1}\out{\lambda_i}{\lambda_i}X(\lambda_j+t)^{-1}\out{\lambda_j}{\lambda_j}dt.
\end{eqnarray*}
Finally,
\begin{eqnarray*}
\mathbb{K}^{-1}_\rho(X) &=&\sum_{i,j}
\mathbb{L}_{\out{\lambda_i}{\lambda_i}}\mathbb{R}_{\out{\lambda_j}{\lambda_j}}\mathbb{K}^{-1}_\rho(X)\\
&=&\sum_{i,j}\int^\infty_0(\lambda_i+t)^{-1}\out{\lambda_i}{\lambda_i}X(\lambda_j+t)^{-1}\out{\lambda_j}{\lambda_j}dt\\
&=&\int^\infty_0 \Pa{\sum_i(\lambda_i+t)^{-1}\out{\lambda_i}{\lambda_i}}X\Pa{\sum_j(\lambda_j+t)^{-1}\out{\lambda_j}{\lambda_j}}dt\\
&=& \int^\infty_0(\rho+t)^{-1}X(\rho+t)^{-1}dt.
\end{eqnarray*}
In what follows, we show that
$$
\mathbb{K}_\rho\Set{C^\dagger=C: \Tr{\rho C}=0}=\Set{B^\dagger=B:
\Tr{B}=0}.
$$
Since $\Tr{\mathbb{K}_\rho(C)}=\Tr{\rho C}$, it follows that
$$
\mathbb{K}_\rho\Set{C^\dagger=C: \Tr{\rho C}=0} \subseteq
\Set{B^\dagger=B: \Tr{B}=0}.
$$
Now let $B\in\Set{B^\dagger=B: \Tr{B}=0}$. Since $\mathbb{K}_\rho$
is invertible, the equation $B=\mathbb{K}_\rho(X)$ has a unique
solution: $X=\mathbb{K}^{-1}_\rho(B)$. It suffice to show $\Tr{\rho
X}=0$. Clearly
\begin{eqnarray*}
\Tr{\rho X} &=& \Tr{\rho\mathbb{K}^{-1}_\rho(B)} = \Tr{\rho
\int^\infty_0(\rho+t)^{-1}B(\rho+t)^{-1}dt}\\
&=& \int^\infty_0 \Tr{\rho(\rho+t)^{-2}B}dt = \sum_i\int^\infty_0
\frac{\lambda_i}{(\lambda_i+t)^2}dt
\Innerm{\lambda_i}{B}{\lambda_i}\\
&=& \sum_i \Innerm{\lambda_i}{B}{\lambda_i} = \Tr{B} = 0.
\end{eqnarray*}
\end{exam}

\end{document}

%% file: preamble.tex
\usepackage[T1]{fontenc}
\usepackage[sc]{mathpazo}
\usepackage{amsmath}
\usepackage{amssymb}
\usepackage{enumerate}
\usepackage{amsthm}
\usepackage{amsfonts,mathrsfs}
\usepackage[style]{fncychap}
\usepackage{graphicx, color} 
\usepackage{geometry} 
\usepackage{eepic}
\usepackage{ifthen}

\usepackage{subfigure}
\usepackage{multirow}
\newboolean{ElectronicVersion}
\setboolean{ElectronicVersion}{true} 

\usepackage[unicode=true,pdfusetitle, bookmarks=true,bookmarksnumbered=false,bookmarksopen=false, breaklinks=false,pdfborder={0 0 0},backref=false,colorlinks=false] {hyperref}
\hypersetup{
colorlinks,linkcolor=myurlcolor,citecolor=myurlcolor,urlcolor=myurlcolor}
\definecolor{myurlcolor}{rgb}{0,0,0.7}

\usepackage{hyperref}
\hypersetup{pdfpagemode=UseNone}

\newcommand{\tinyspace}{\mspace{1mu}}

\newcommand{\op}[1]{\operatorname{#1}}

\newcommand{\abs}[1]{\left\lvert\tinyspace #1 \tinyspace\right\rvert}

\newcommand{\norm}[1]{\left\lVert\tinyspace #1 \tinyspace\right\rVert}

\renewcommand{\t}{{\scriptscriptstyle\mathsf{T}}}

\newcommand{\setft}[1]{\mathrm{#1}}
\newcommand{\lin}[1]{\setft{L}\left(#1\right)}
\newcommand{\density}[1]{\setft{D}\left(#1\right)}

\newcommand{\trans}[1]{\setft{T}\left(#1\right)}

\newcommand{\pd}[1]{\setft{Pd}\left(#1\right)}

\renewcommand{\vec}{\op{vec}}

\newcommand{\rank}{\op{rank}}

\newcommand{\sign}{\op{sign}}

\def\complex{\mathbb{C}}
\def\real{\mathbb{R}}
\def\natural{\mathbb{N}}

\def\I{\mathbb{1}}

\newenvironment{mylist}[1]{\begin{list}{}{
    \setlength{\leftmargin}{#1}
    \setlength{\rightmargin}{0mm}
    \setlength{\labelsep}{2mm}
    \setlength{\labelwidth}{8mm}
    \setlength{\itemsep}{0mm}}}
    {\end{list}}

\def\ot{\otimes}

\newcommand{\inner}[2]{\langle #1 , #2\rangle}
\newcommand{\iinner}[2]{\langle #1 | #2\rangle}
\newcommand{\out}[2]{| #1\rangle\langle #2 |}
\newcommand{\Inner}[2]{\left\langle #1 , #2\right\rangle}
\newcommand{\Innerm}[3]{\left\langle #1 \left| #2 \right| #3 \right\rangle}
\newcommand{\defeq}{\stackrel{\smash{\textnormal{\tiny def}}}{=}}

\newcommand{\pa}[1]{(#1)}
\newcommand{\Pa}[1]{\left(#1\right)}

\newcommand{\Br}[1]{\left[#1\right]}
\newcommand{\set}[1]{\{#1\}}
\newcommand{\Set}[1]{\left\{#1\right\}}

\newcommand{\bra}[1]{\langle#1|}

\newcommand{\ket}[1]{|#1\rangle}

\DeclareMathOperator{\vectorize}{vec}
\newcommand{\col}[1]{\vectorize\pa{#1}}

\DeclareMathOperator{\trace}{Tr}

\newcommand{\Ptr}[2]{\trace_{#1}\Pa{#2}}

\newcommand{\Tr}[1]{\Ptr{}{#1}}

\def\cA{\mathcal{A}}\def\cB{\mathcal{B}}\def\cC{\mathcal{C}}
\def\cH{\mathcal{H}}\def\cI{\mathcal{I}}
\def\cM{\mathcal{M}}\def\cN{\mathcal{N}}
\def\cP{\mathcal{P}}
\def\cX{\mathcal{X}}\def\cY{\mathcal{Y}}

\def\bH{\mathbf{H}}

\def\bP{\mathbf{P}}\def\bS{\mathbf{S}}

\def\rH{\mathrm{H}}

\def\rR{\mathrm{R}}\def\rS{\mathrm{S}}

\def\sH{\mathscr{H}}

\newcommand{\cmp}{Comm. Math. Phys.~}

\newcommand{\jmp}{J. Math. Phys.~}
\newcommand{\jpa}{J. Phys. A~}
\newcommand{\lmp}{Lett. Math. Phys.~}

\newcommand{\prl}{Phys. Rev. Lett.~}

\newcommand{\rmp}{Rev. Math. Phys.~}

\newtheorem{thrm}{Theorem}[section]
\newtheorem{lem}[thrm]{Lemma}
\newtheorem{prop}[thrm]{Proposition}
\newtheorem{cor}[thrm]{Corollary}

\theoremstyle{definition}
\newtheorem{definition}[thrm]{Definition}
\newtheorem{remark}[thrm]{Remark}
\newtheorem{exam}[thrm]{Example}
\numberwithin{equation}{section}

\newcounter{questionnumber}